%% file: SPLextension.tex
\documentclass[journal,10pt,twocolumn,oneside]{IEEEtran}
\usepackage{array}
\usepackage{changepage}
\ifCLASSINFOpdf
   \usepackage[pdftex]{graphicx}
\else
\fi
\usepackage{bm}
\usepackage{graphicx}
\usepackage{subfigure}

\usepackage{multirow}
\usepackage{multicol}
\usepackage{cite}
\usepackage{epstopdf}
\usepackage[cmex10]{amsmath}
\usepackage{amsthm}
\usepackage{amsmath}
\usepackage{framed}
\newtheorem{lemma}{\textbf{Lemma}}
\newtheorem{theorem}{\textbf{Theorem}}

\newtheorem{proposition}{\textbf{Proposition}}
\newtheorem{definition}{\textbf{Definition}}
\usepackage{algorithm} 
\usepackage{algorithmic} 
\newcommand*\varhrulefill[1][0.4pt]{\leavevmode\hrule height#1\kern0pt}

\usepackage{enumerate}

\usepackage{changepage}

\usepackage{mdwmath}
\usepackage{makecell,rotating,multirow,diagbox}

\usepackage{multirow}
\usepackage{booktabs}
\usepackage{dcolumn}
\usepackage{amssymb}

\usepackage{url}

\usepackage{subfloat}
\usepackage{color}
\usepackage{times,amsmath,color,amssymb,epsfig,cite,subfigure,algorithm,algorithmic}
\usepackage{flushend}
\usepackage{fancyhdr}
\fancypagestyle{plain}{
    \fancyhf{}
    \fancyfoot[C]{\textsc{978-1-5090-0690-8/16/\$31.00 \copyright 2016 IEEE}}
    
    }
    \fancypagestyle{empty}{
    \fancyhf{}
    \fancyfoot[C]{\textsc{}}
    
    }

\graphicspath{{figures/}}

\begin{document}
\title{Low-complexity Graph Sampling with Noise and Signal Reconstruction via Neumann Series}
\author{
\IEEEauthorblockN{Fen Wang, {Gene Cheung, \emph{Senior Member, IEEE}}, and Yongchao Wang, \emph{Member, IEEE}}
\renewcommand{\baselinestretch}{1.0}
\thanks{This work was supported in part by the National Science Foundation of
China under Grant 61771356 and in part by the 111 project of China under
Grant B08038. \emph{(Corresponding author: Yongchao Wang.)}}
\thanks{F. Wang and Y. Wang are with State Key Laboratory of ISN, Xidian University, Xi'an 710071, Shaanxi, China (e-mail: fenwang@stu.xidian.edu.cn; ychwang@mail.xidian.edu.cn).}
\thanks{{G. Cheung is with the department of EECS, York University, 4700 Keele Street, Toronto, M3J 1P3, Canada (e-mail:genec@yorku.ca).}}
}
\maketitle
\thispagestyle{empty} 

\vspace{0.01cm}
\begin{abstract}
\input{abstract}
\end{abstract}

\vspace{-0.1in}
\section{Introduction}
\label{sec:intro}
\input{intro}


\vspace{-0.05in}
\section{Preliminaries}
\label{sec:formulate}
\input{formulate}

\vspace{-0.05in}
\section{Augmented A-optimality Graph Signal Sampling}
\label{sec:AOptimal}
\input{aoptimal}
\vspace{-0.05in}
\section{Fast Greedy Sampling}
\label{sec:greedy}
\input{greedy}
\vspace{-0.05in}
\section{ Dynamic Subset Sampling}
\label{sec:DSS}
\input{NodeExchange}
\vspace{-0.05in}
\section{Accompanied Graph Signal Reconstruction}
\label{sec:reconstruction}
\input{reconstruction}

\vspace{-0.05in}
\section{Experimental Results}
\label{sec:experiments}
\input{experiments}
\vspace{-0.05in}
\section{Conclusion}
\label{sec:conclusion}
\input{conclusion}

\appendices
\section{Proof of Proposition \ref{infinite series}}
\label{eigenScope}
From the definition of $\mathbf{C}$, ${\mathbf{C}^{\top}}\mathbf{C} = {\left[ \begin{array}{l}
 {\mathbf{I}_\mathcal{S}}~~~{\bf{0}} \\
 {\bf{0}}~~~~{\kern 1pt}{\bf{0}} \\
 \end{array} \right]}$ under appropriate permutation.
Then, $\forall \mathbf{x} \in {\mathbb{R}^K}$ and $\left\| \mathbf{x} \right\|_2 = 1$,
\begin{equation}\label{product}
\begin{split}
\begin{array}{l}
 {\mathbf{x}^{\top}}\mathbf{\Psi} \mathbf{x} = {\left( {{\mathbf{V}_K}\mathbf{x}} \right)^{\top}}\left( {{\mathbf{C}^{\top}}\mathbf{C}} \right)\left( {{\mathbf{V}_K}\mathbf{x}} \right) \\
~~~~~~~~= {\mathbf{b}^{\top}}\left[ \begin{array}{l}
 {\mathbf{I}_\mathcal{S}}~~~{\bf{0}} \\
 {\bf{0}}~~~~ {\kern 1pt} {\bf{0}} \\
 \end{array} \right]\mathbf{b}={\mathbf{b}^{\top}_{{\mathcal{S}}}}{\mathbf{b}_{{\mathcal{S}}}}, \\
 \end{array}
 \end{split}
\end{equation}
{where $\mathbf{b}={{\mathbf{V}_K}\mathbf{x}}$}.

Since ${\mathbf{b}^{\top}}\mathbf{b} = {\left( {{\mathbf{V}_K}\mathbf{x}} \right)^{\top}}\left( {{\mathbf{V}_K}\mathbf{x}}\right)=1$, $0 \le {\mathbf{x}^{\top}}\mathbf{\Psi} \mathbf{x} \le 1$.
Because $\textrm{rank}{{{\left( \mathbf{C}{\mathbf{V}_K} \right)}}} = K$, $\mathbf{\Psi}$ is positive definite, which results in $0 < {\mathbf{x}^{\top}}\mathbf{\Psi} \mathbf{x} \le 1$ and $0 \le {\mathbf{x}^{\top}}\mathbf{\Phi} \mathbf{x} < 1$.
Due to the {Rayleigh quotient theorem}, $0 \le {\delta _i}<1$.
\section{Proof of Lemma \ref{submodular}}
\label{submodularProof}
If we view matrix $\tilde{\mathbf{V}}_K$ as $\tilde{\mathbf{V}}_K=[\mathbf{t}_1~\mathbf{t}_2~...\mathbf{t}_N]^{\top}$ with $\mathbf{t}_i\in \mathbb{R}^{K}$, then $(\mathbf{C}\tilde{\mathbf{V}}_K)^{\top}\mathbf{C}\tilde{\mathbf{V}}_K=\sum_{i\in \mathcal{S}}\mathbf{t}_i\mathbf{t}^{\top}_i$.
Define $\mathbf{Z}(\mathcal{S})=\sum_{i\in \mathcal{S}}\mathbf{t}_i\mathbf{t}^{\top}_i+\mu\mathbf{I}$, we know that  $\mathbf{Z}\left(\mathcal{S}\cup\{j\}\right)=\mathbf{Z}(\mathcal{S})+\mathbf{t}_j\mathbf{t}^{\top}_j$ for any $j\not\in\mathcal{S}$ and $g(\mathcal{S})=\text{tr}[\mathbf{Z}(\mathcal{S})^{-1}]$.

(i)~Monotonic decreasing
\begin{equation}
\begin{split}
g(\mathcal{S}\cup\{j\})-g(\mathcal{S})=\text{tr}\left[(\mathbf{Z}(\mathcal{S})+\mathbf{t}_j\mathbf{t}^{\top}_j)^{-1}\right]
-\text{tr}[\mathbf{Z}(\mathcal{S})^{-1}]\nonumber
\end{split}
\end{equation}
for any $j\not\in\mathcal{S}$.

From Lemma \ref{matrix lemma2}, we know that
\begin{equation}
\begin{split}
  &  g(\mathcal{S}\cup\{j\})-g(\mathcal{S})=-\text{tr}\left[\frac{\mathbf{Z}(\mathcal{S})^{-1}
    \mathbf{t}_j\mathbf{t}^{\top}_j\mathbf{Z}(\mathcal{S})^{-1}}
    {1+\mathbf{t}^{\top}_j\mathbf{Z}(\mathcal{S})^{-1}\mathbf{t}_j}\right]\\
 & \hspace{2.8cm}  =-\frac{\|\mathbf{Z}(\mathcal{S})^{-1}
    \mathbf{t}_j\|^2_2}
    {1+\mathbf{t}^{\top}_j\mathbf{Z}(\mathcal{S})^{-1}\mathbf{t}_j}
\end{split}
\end{equation}

It is easy to prove the eigenvalues of $\mathbf{Z}(\mathcal{S})$ are in $(\mu,1+\mu]$ for any $\mathcal{S}$ from the similar derivation in Proposition 1 \cite{SPL}.
Therefore, $\mathbf{t}^{\top}_j\mathbf{Z}(\mathcal{S})^{-1}\mathbf{t}_j\geq
\lambda^{-1}_{\text{max}}[\mathbf{Z}(\mathcal{S})]\cdot\|\mathbf{t}_j\|^2_2\geq\frac{\|\mathbf{t}_j\|^2_2}{1+\mu}>0$ from {Rayleigh quotient theorem}.
Using this result and combining it with the above two equations, we know that
$g(\mathcal{S}\cup\{j\})-g(\mathcal{S})\leq0$ for any $j\not\in\mathcal{S}$ which implies the set function $g$ is monotonic decreasing.

(ii) $\alpha$-supermodular

We first present the definition of $\alpha$-supermodularity introduced in paper \cite{greedybound}.
\begin{definition}
A set function $g:~2^{\mathcal{V}}\rightarrow \mathbb{R}$ is $\alpha$-supermodular if for all sets $\mathcal{A}\subseteq\mathcal{B}\subseteq\mathcal{V}$ and all $j\not\in\mathcal{B}$, the following equation holds for some $\alpha\geq0$
\begin{equation}\label{alphaSupermodular}
    g\left(\mathcal{A}\cup\{j\}\right)-g(\mathcal{A})\leq \alpha\left[ g\left(\mathcal{B}\cup\{j\}\right)-g(\mathcal{B})\right]
\end{equation}
\end{definition}

$\alpha$-supermodularity is only of interest when $\alpha$ takes the largest value \cite{greedybound}:
\begin{equation}\label{minimalValue}
\begin{split}
&\hspace{0cm}   \alpha=\mathop{\text{min}}\limits_{\mathop{\mathcal{A}\subseteq\mathcal{B}\subseteq\mathcal{V}}\limits
    _{j\not\in\mathcal{B}}}
    \frac{g\left(\mathcal{A}\cup\{j\}\right)-g(\mathcal{A})}{g\left(\mathcal{B}\cup\{j\}\right)-g(\mathcal{B})}\\
&\hspace{0.3cm}
=\mathop{\text{min}}\limits_{\mathop{\mathcal{A}\subseteq\mathcal{B}\subseteq\mathcal{V}}\limits
    _{j\not\in\mathcal{B}}}
    \frac{1+\mathbf{t}^{\top}_j\mathbf{Z}(\mathcal{B})^{-1}\mathbf{t}_j}
    {1+\mathbf{t}^{\top}_j\mathbf{Z}(\mathcal{A})^{-1}\mathbf{t}_j}\cdot
    \frac{\mathbf{t}^{\top}_j\mathbf{Z}(\mathcal{A})^{-2}\mathbf{t}_j}
    {\mathbf{t}^{\top}_j\mathbf{Z}(\mathcal{B})^{-2}\mathbf{t}_j}
\end{split}
\end{equation}

Since $\lambda^{-1}_{\text{max}}[\mathbf{Z}(\mathcal{S})]\|\mathbf{t}_j\|^2_2
\leq\mathbf{t}^{\top}_j\mathbf{Z}(\mathcal{S})^{-1}\mathbf{t}_j\leq
\lambda^{-1}_{\text{min}}[\mathbf{Z}(\mathcal{S})]\|\mathbf{t}_j\|^2_2$, the lower bound of $\alpha$ is
\begin{equation}\label{alphaLB}
    \alpha\geq \frac{1+\lambda^{-1}_{\text{max}}[\mathbf{Z}(\mathcal{B})]\|\mathbf{t}_j\|^2_2}
    {1+\lambda^{-1}_{\text{min}}[\mathbf{Z}(\mathcal{A})]\|\mathbf{t}_j\|^2_2}\cdot
    \frac{\lambda^{-2}_{\text{max}}[\mathbf{Z}(\mathcal{A})]\|\mathbf{t}_j\|^2_2}
    {\lambda^{-2}_{\text{min}}[\mathbf{Z}(\mathcal{B})]\|\mathbf{t}_j\|^2_2}
\end{equation}

As we claimed, the eigenvalues of $\mathbf{Z}(\mathcal{S})$ are in $(\mu,1+\mu]$, so
\begin{equation}\label{LBwith mu}
    \alpha\geq \frac{1+(1+\mu)^{-1}\|\mathbf{t}_j\|^2_2}{1+\mu^{-1}\|\mathbf{t}_j\|^2_2}\cdot
    \frac{(1+\mu)^{-2}}{\mu^{-2}}\doteq\alpha'
\end{equation}

It is easy to drive that
\begin{equation}
    \frac{\partial\alpha'}{\partial\|\mathbf{t}_j\|^2_2}=\frac{(1+\mu)^{-2}}{\mu^{-2}}
    \cdot\frac{(1+\mu)^{-1}-{\mu}^{-1}}{[1+{\mu}^{-1}\|\mathbf{t}_j\|^2_2]^{2}}<0
\end{equation}
which demonstrates the function $\alpha'$ is decreasing in terms of $\|\mathbf{t}_j\|^2_2$.
Since $\mathbf{t}_j$ is a row of $\mathbf{V}_K$, \textit{i.e.}, $\|\mathbf{t}_j\|^2_2\leq1$, we have
\begin{equation}
\alpha\geq\frac{[1+(1+\mu)^{-1}](1+\mu)^{-2}}{[1+\mu^{-1}]\mu^{-2}}=\frac{\mu^3(\mu+2)}{(\mu+1)^4}>0
\end{equation}
Hence, there exists $\alpha>0$ to fulfill equation \eqref{alphaSupermodular} and its value can be bounded by $\frac{\mu^3(\mu+2)}{(\mu+1)^4}$, which means this function is approximately supermodular with parameter $\alpha$.

\section{Proof of Proposition \ref{robustness}}
\label{robustnessProof}
Assume the eigen-decomposition of $\mathbf{\Psi}$ is $\mathbf{\Psi}=\mathbf{U} \mathbf{\Lambda} {{\mathbf{U}^{\top}}}$
where $\mathbf{\Lambda}=\textrm{diag}\{\sigma _i\}$ and $\mathbf{U}\mathbf{U}^{\top}=\mathbf{I}$.
We have derived in Appendix.B in the paper \cite{SPL} that
\begin{align}
&\hspace{0cm} {\mathbb{E}}\left\| {{\bf{\hat x}_{\textrm{LS}}} - {\bf{x}}} \right\|_2^2
= {\omega ^2}\sum\limits_{i = 1}^K {\frac{1}{{1 - {\delta _i}}}}= {\omega ^2}\sum\limits_{i = 1}^K {\frac{1}{{ {\sigma _i}}}}.
\end{align}

Rewrite the GFS reconstruction $\hat{\mathbf{x}} = {\mathbf{V}_K}[ {{{(\mathbf{C}{\mathbf{V}_K})}^{\top}}\mathbf{C}{\mathbf{V}_K}}+\beta\mathbf{I} ]^{-1}{(\mathbf{CV}_K)^{\top}}\mathbf{y}_\mathcal{S}$ by $\hat{\mathbf{x}}=\mathbf{H}\mathbf{y}_\mathcal{S}=\mathbf{H}\mathbf{C}(\mathbf{x}+\mathbf{n})$, then the bias of this estimator is
\begin{align}
&\hspace{0cm} \textrm{Bias}({\bf{\hat x}};\beta)=\left\| \mathbb{E}[\bf{\hat x}]-\bf{x} \right\|_2
= \left\| (\bf{HC}-\bf{I})\bf{x} \right\|_2
\end{align}

Recall that the original graph signal was assumed $K$-bandlimited, so
\begin{equation}
\begin{split}
&\hspace{0cm}\textrm{Bias}^2({\bf{\hat x}};\mu)=\left\| (\mathbf{HCV}_K-\mathbf{V}_K)\tilde{\mathbf{x}}_K\right\|^2_2\\
&\hspace{0.5cm}=\tilde{\mathbf{x}}^{\top}_K(\mathbf{HCV}_K-\mathbf{V}_K)^{\top}(\mathbf{HCV}_K-\mathbf{V}_K)\tilde{\mathbf{x}}_K\\
\end{split}
\end{equation}
in which
\begin{equation}
\begin{split}
&\hspace{0cm}(\mathbf{HCV}_K-\mathbf{V}_K)^{\top}(\mathbf{HCV}_K-\mathbf{V}_K)\\
&\hspace{0cm}=\left[\mathbf{V}_K(\mathbf{\Psi}+\beta\mathbf{I})^{-1}\mathbf{\Psi}-\mathbf{V}_K\right]^{\top}
\left[\mathbf{V}_K(\mathbf{\Psi}+\beta\mathbf{I})^{-1}\mathbf{\Psi}-\mathbf{V}_K\right]\\
&\hspace{0cm}=\mathbf{\Psi}(\mathbf{\Psi}+\beta\mathbf{I})^{-2}\mathbf{\Psi}-\mathbf{\Psi}(\mathbf{\Psi}
+\beta\mathbf{I})^{-1}
-(\mathbf{\Psi}+\beta\mathbf{I})^{-1}\mathbf{\Psi}+\mathbf{I}\\
&\hspace{0cm}=\mathbf{U}\textrm{diag}\left(\frac{\sigma^2_i}{(\sigma_i+\beta)^2}-\frac{2\sigma_i}{\sigma_i+\beta}
+1\right)\mathbf{U}^{\top}\\
&\hspace{0cm}=\mathbf{U}\textrm{diag}\left[{\left(1+\frac{\sigma_i}{\beta}\right)^{-2}}\right]\mathbf{U}^{\top}\\
&\hspace{0cm}=\sum^K_{i=1}\left(1+\frac{\sigma_i}{\beta}\right)^{-2}\mathbf{u}_i\mathbf{u}^{\top}_i\nonumber
\end{split}
\end{equation}

Therefore,
\begin{equation}
\begin{split}
&\hspace{0cm}\textrm{Bias}^2({\bf{\hat x}};\beta)=\sum^K_{i=1}\left(1+\frac{\sigma_i}{\beta}\right)^{-2}(\mathbf{u}^{\top}_i\tilde{\mathbf{x}}_K)^2
\end{split}
\end{equation}

The variance part of this solution can be derived from its covariance matrix: $\text{cov}(\mathbf{\hat{x}})=\mathbb{E}\left[(\mathbf{\hat{x}}-\mathbb{E}[\mathbf{\hat{x}}])(\mathbf{\hat{x}}-\mathbb{E}[\mathbf{\hat{x}}])^{\top}\right]
=\mathbb{E}\left[\mathbf{HCnn}^{\top}\mathbf{C}^{\top}\mathbf{H}^{\top}\right]$. Because $\mathbb{E}[\mathbf{nn}^{\top}]=\omega^2\bf{I}$ and $\mathbf{C}\mathbf{C}^{\top}=\mathbf{I}$, the variance of $\mathbf{\hat{x}}$ is
\begin{align}\label{variance}
&\hspace{0cm} \textrm{Var}({\mathbf{\hat x}};\beta)=\text{tr}[\text{cov}(\mathbf{\hat{x}})]
=\omega^2\text{tr}[\bf{HH}^{\top}]\nonumber
\end{align}

From the definition of $\bf{H}$ and the property of trace operator,
\begin{equation}\label{variance}
\begin{split}
&\hspace{0cm}\text{tr}[\mathbf{HH}^{\top}]=\text{tr}\left[{\mathbf{V}_K}( \pmb{\Psi}+\beta\mathbf{I} )^{-1}\pmb{\Psi}( \pmb{\Psi}+\beta\mathbf{I} )^{-1}{\mathbf{V}^{\top}_K}\right]\\
&\hspace{1.3cm}=\text{tr}\left[\mathbf{U} \mathbf{\Lambda} {{\mathbf{U}^{\top}}}\mathbf{U} (\mathbf{\Lambda}+\beta\mathbf{I})^{-2} {{\mathbf{U}^{\top}}}\right]\\
&\hspace{1.3cm}=\text{tr}\left[\mathbf{U}\text{diag}\{\frac{\sigma_i}{(\sigma_i+\beta)^2}\}  {{\mathbf{U}^{\top}}}\right]\\
&\hspace{1.3cm}=\sum^{K}_{i=1}\frac{\sigma_i}{(\sigma_i+\beta)^2}
\end{split}
\end{equation}

Therefore, the variance of $\mathbf{\hat{x}}$ is
\begin{align}\label{variance}
&\hspace{0cm} \textrm{Var}({\bf{\hat x}};\mu)=\omega^2\sum^{K}_{i=1}\frac{\sigma_i}{(\sigma_i+\beta)^2}
\end{align}

\section{Recognize the Bad Initial Available Subset}
\label{badSET}
Space $PW_\eta(G)$ contains all bandlimited graph signals with bandwidth less than $\eta$ and space $L_2(\mathcal{S}^c)$ covers graph signals fulfilling $\mathbf{f}(\mathcal{S})=0$ \cite{Pesenson2008}.
Authors in \cite{AO} claimed that for set $\mathcal{S}$ to be a uniqueness set of space $PW_\eta(G)$, $\eta$ needs to be less than the minimum bandwidth of all signals in space $L_2(\mathcal{S}^c)$, which is defined as the \emph{cutoff frequency} of $L_2(\mathcal{S}^c)$.
Thus, a set $\mathcal{S}$ to be good means its {cutoff frequency} should be as large as possible.
They defined the {cutoff frequency} of space $L_2(\mathcal{S}^c)$ as
\begin{equation}
\eta_c(\mathcal{S})=\mathop{\text{min}}\limits_{\mathop{\mathbf{f}\in L_2(\mathcal{S}^c)}\limits_{\mathbf{f}\not=\mathbf{0}}} \eta(\mathbf{f})
\end{equation}
where $\eta(\mathbf{f})$ is the frequency of a specific graph signal $\mathbf{f}$.
%

A strategy was proposed in \cite{AO} to approximate the cutoff frequency by
\begin{equation}
\begin{split}
&\hspace{0cm}
\Omega_k(\mathcal{S})\doteq\mathop{\text{min}}\limits_{\mathop{\mathbf{f}\in L_2(\mathcal{S}^c)}\limits_{\mathbf{f}\not=\mathbf{0}}} \left(\frac{\|\mathbf{L}^k\mathbf{f}\|}{\|\mathbf{f}\|}\right)^{1/k}\\
&\hspace{0cm}
=\left[\mathop{\text{min}}\limits_{\pmb{\psi}}
\frac{\pmb{\psi}^{\top}\left((\mathbf{L}^{\top})^k\mathbf{L}^k\right)_{\mathcal{S}^c}\pmb{\psi}}
{\pmb{\psi}^{\top}\pmb{\psi}}\right]^{1/{2k}}=(\lambda_{1,k})^{1/{2k}}
\end{split}
\end{equation}
where $\lambda_{1,k}$ is the smallest eigenvalue of the reduced matrix $\left((\mathbf{L}^{\top})^k\mathbf{L}^k\right)_{\mathcal{S}^c}$.

This cutoff frequency can be used for judging if a set is good for initialization because a reasonable
set won't have too small $\Omega_k(\mathcal{S})$.
We take this as a criterion to detour a bad initial subset.
During experiments, the size of the available set is $M=800$, so we randomly pick 50 sets off-line with $M=800$ to compute their approximate cutoff frequency $\Omega_k(\mathcal{S})$ and sort these values in ascending order.
We take the fifth value in this rank, denoted by $\tilde{\Omega}$, as the bound for define a bad set.
For an initial set $\mathcal{S}$,  once $\Omega_k(\mathcal{S})>\tilde{\Omega}$, we define it to be a good set and then perform sampling on it.
The following available sets are generated from this initial set with small crossover probability, so their reliability can be guaranteed.
\input{reference}
\end{document}

%% file: abstract.tex
Graph sampling addresses the problem of selecting a node subset in a graph to collect samples, so that a $K$-bandlimited signal can be reconstructed in high fidelity.
Assuming an independent and identically distributed (i.i.d.) noise model, minimizing the expected mean square error (MMSE) leads to the known A-optimality criterion for graph sampling, which is expensive to compute and difficult to optimize.
In this paper, we propose an augmented objective based on Neumann series that well approximates the original criterion and is amenable to greedy optimization.
Specifically, we show that a shifted A-optimal criterion can be equivalently written as {a function of} an ideal low-pass (LP) graph filter, which in turn can be approximated efficiently via fast graph Fourier transform (FGFT).
Minimizing the new objective, we select nodes greedily without large matrix inversions using a matrix inverse lemma. 
Further, for the dynamic network case where node availability varies across time, we propose an extended sampling strategy that replaces offline samples one-by-one in the selected set.
For signal reconstruction, we propose an accompanied biased signal recovery strategy that reuses the approximated filter from sampling.
Experiments show that our reconstruction is more robust to large noise than the least square (LS) solution, and our sampling strategy far outperforms several existing schemes. 

%% file: intro.tex
\textit{Graph signal processing} (GSP) is the study of signals that reside on irregular data kernels described by graphs \cite{GSP,emerging,GSIP}.
To analyze graph signals spectrally, a large amount of research strived to define frequencies on graphs, and subsequently design transforms and wavelets for spectral decomposition of signals \cite{frequency,wavelet,filterbank}.
Specifically, \textit{graph Fourier transform} (GFT) of a graph signal $\mathbf{x} \in \mathbb{R}^N$ is defined by its projection on the eigenvector space of the graph Laplacian matrix \cite{emerging} (or the adjacency matrix \cite{gsp})
\footnote{\cite{Irregularity} proposed an alternative definition of generalized GFT that takes the irregularity of the graph into account.
For direct graphs, \cite{direct-graph} proposed the notion of graph direct variation.
In this paper, we focus on GFT derived from the Laplacian operator, but our sampling method is applicable to other symmetrical graph operators.}.
A graph signal of \textit{bandwidth} $K$ is a signal with non-zeros GFT coefficients for $K$ eigenvectors associated with the $K$ smallest eigenvalues.
Sampling of bandlimited graph signals is an important basic problem in GSP, since sensing in practice is often expensive, \textit{e.g.}, { SAR images from remote sensing \cite{remote sensing} and label assignment of a medical examination \cite{medical}.}
Formally, graph sampling is the study of node selection in a graph to collect samples, such that the original signal can be reconstructed in high fidelity.

Previous works in graph sampling can be broadly divided into two categories: noiseless sampling and noisy sampling.
Given that samples are noiselessly observed, \cite{cornell} showed that random sampling can result in perfect signal reconstruction with high probability, assuming that the sample size is larger than the signal's bandwidth.
However, the condition number of the corresponding coefficient matrix can be large, resulting in an unstable signal reconstruction.
In response,  \cite{towards,AO} proved a necessary and sufficient condition of a
\emph{uniqueness set} \cite{Pesenson2008},
and then proposed a sampling strategy with simple matrix-vector computations based on the notion of spectral proxies (SP).

If the observed samples are corrupted with noise, one classical criterion is to optimize a \textit{minimum mean square error} (MMSE) function, which leads to an A-optimality design criterion assuming an independent and identically distributed (i.i.d.) additive noise model \cite{Boyd}.
Minimizing the A-optimality criterion for graph sampling is known to be difficult: evaluating the criterion for a fixed sample set requires computation of eigenvectors and matrix inverse, and optimizing the sample set for a fixed budget is combinatorial.

To avoid heavy computation,  \cite{Uncertainty} used a greedy procedure to optimize the MMSE problem directly, called \emph{minimum Frobenius norm} (MFN).
Further, \cite{greedybound} proved that the performance of the greedy solution can be bounded via super-modularity analysis.
Alternatively, \cite{SCsampling} used the E-optimality criterion for graph sampling, which can be interpreted as the worst case MSE.
Nonetheless, all proposals above required eigen-decomposition of the Laplacian operator and {costly computations in each greedy step, such as matrix inverse for MFN and singular value decomposition (SVD) for E-optimal sampling}.

To reduce the complexity of evaluating the objective, an eigen-decomposition-free sampling strategy was recently proposed in \cite{sensorSelection,eigenfree} by constraining the coverage of a localization operator, implemented using Chebyshev polynomial approximation
\footnote{{Our proposed method is a function of an ideal graph low-pass filter, which can be viewed as an ideal operator-based sampling defined in \cite{eigenfree}. However, we do not implement this filter by a localized operator. }}.
{However, the performance of this method is competitive only when the sample size is smaller than signal's bandwidth.}
Orthogonally, our previous work,  called \textit{matrix inversion approximation} (MIA) \cite{SPL}, developed a sampling strategy to optimize an approximate MSE function based on truncated Neumann series, but the number of terms in the truncated sum must be sufficiently large for the approximation error to be small, limiting its practicality.

Because noisy graph sampling problem is NP-hard in nature \cite{greedybound}, all aforementioned criteria are optimized in a greedy manner.
To lower complexity, random sample node selection was proposed to sample signals based on a specified probability \cite{randomSampling,structured sampling}.
However, for the same sampling budget it cannot guarantee stable performance comparable to the deterministic counterparts \cite{GSP,eigenfree}.

In this paper, we propose a novel low-complexity deterministic sampling strategy via Neumann series expansion.
Unlike previous approaches {\cite{Uncertainty,SCsampling,eigenfree}}, we minimize a variant of the A-optimality criterion directly but without eigen-decomposition or matrix inversion.
Yet different from our work in \cite{SPL}, there is no explicit computation of the Neumann series sum in our augmented objective, and thus difficult truncation / approximation tradeoff is not necessary.

Specifically, we first propose an augmented objective to approximate the A-optimality criterion, and prove the inverse of its information matrix is equal to a matrix series based on the Neumann series theorem \cite{neumannseries}.
We then rewrite the objective as a function of an ideal low-pass (LP) graph filter with cutoff frequency $K$, which is efficiently approximated using fast Graph Fourier Transform (FGFT) \cite{France Givens}.
To optimize the objective function, we select nodes greedily without any matrix inversions based on a matrix inverse lemma.
Further, the greedy solution is proved to have upper-bounded performance based on super-modularity analysis \cite{greedybound}.

All previous graph sampling works assumed a static node set, while in practice often the availability of nodes as sensors varies over time; \textit{e.g.}, one person goes offline in a social network, or one sensor runs {out of battery} in a sensor network \cite{sensorBroken}.
We additionally address the sampling problem for dynamic time-varying node subsets, where the availability of each node varies as a function of time.
Specifically, we develop an extended sampling strategy that replaces offline samples one-by-one in the selected set via Sherman-Morrison formula \cite{matrix lemma} to reduce complexity.
\emph{To the best of our knowledge, we are the first to formulate and tackle the sampling problem on dynamic node subsets in the literature. }

Finally, we develop an accompanied signal reconstruction strategy to reuse the LP graph filter from sampling.
This reconstruction strategy provides a biased estimator with lower complexity, which is more robust to large noise than the conventional least square (LS) solution \cite{AO}.

{To summarize}, our main contributions are as follows:
\begin{enumerate}
\item We propose a novel graph sampling objective based on Neumann series expansion, which well approximates the A-optimality criterion.
\item We design a fast sampling algorithm to minimize the proposed objective with neither eigen-decomposition nor matrix inversions.
\item For the dynamic network case, we propose an extended node replacement strategy to replace offline nodes one-by-one efficiently.
\item We propose an accompanied biased signal recovery method for sampling that is more robust to large noise.
\item Experiments demonstrate that our proposed sampling method outperforms several state-of-the-art sampling strategies in artificial and real-world graphs.
\end{enumerate}

The outline of this paper is as follows.
In Section \ref{sec:formulate}, we overview GSP concepts and terminology used throughout this paper and discuss the A-optimality sampling problem.
We propose the augmented A-optimal sampling criterion in Section \ref{sec:AOptimal}.
In Section \ref{sec:greedy}, we detail our fast graph sampling algorithm.
We discuss the dynamic graph sampling problem in Section \ref{sec:DSS}.
We develop our biased signal reconstruction method in Section \ref{sec:reconstruction}.
Finally, we present experimental results nad conclusion in Section \ref{sec:experiments} and \ref{sec:conclusion}, respectively. 

%% file: formulate.tex
Denote by $\mathcal{G}=(\mathcal{V},\mathcal{E},\mathbf{W})$ a graph with $N$ nodes indexed by $\mathcal{V}=\{ 1,\ldots,N\}$.
$\mathcal{E}$ specifies the set of connected node pairs $(i,j)$, and the $(i,j)$-th entry $w_{i,j}$ of an $N \times N$ \textit{adjacency matrix} $\mathbf{W}$ is the weight of an edge connecting nodes $i$ and $j$ ($w_{i,j} = 0$ if nodes $i$ and $j$ are not connected).
We additionally define a diagonal \textit{degree matrix} $\mathbf{D}$, where $d_{i,i} = \sum_j w_{i,j}$.
In this paper we focus on connected, undirected graphs with no self-loops, and we adopt the symmetric \textit{combinatorial graph Laplacian matrix} $\mathbf{L}=\mathbf{D}-\mathbf{W}$ as the variation operator.
Suppose that the eigen-decomposition of $\mathbf{L}$ is $\mathbf{L}= \mathbf{V} \boldsymbol{\Sigma} {{\mathbf{V}^{\top}}} $, where {$\mathbf{V}={ [\mathbf{v}_1}, ... ,{\mathbf{v}_N} ]$} contains a set of $N$ orthonormal eigenvectors as columns
corresponding to non-decreasing eigenvalues {$0 = {\lambda _1} <{\lambda _2} \le {\lambda _3} \le \ldots \le {\lambda _N}$}.
Then the GFT of a graph signal $\mathbf{x} \in \mathbb{R}^N$ is defined as its expansion on the eigenvector space of $\mathbf{L}$, \textit{i.e.}, $\tilde {\mathbf{x}} = {\mathbf{V}^{\top}}\mathbf{x}$, and the inverse GFT is $\mathbf{x} = \mathbf{V}\tilde {\mathbf{x}}$.
A graph signal is called $K$-\textit{bandlimited} if its GFT coefficients $\tilde {\mathbf{x}}$ are non-zero only for the lowest $K$ frequencies
\footnote{Authors in \cite{greedybound} defined the $\mathcal{K}$-spectrally sparse graph signal by  ${\mathbf{x} = {\mathbf{V}_{\mathcal{K}}}{{\tilde {\mathbf{x}}}_{\mathcal{K}}}}$ which is a generalization of our formulation since $\mathcal{K}\in\mathcal{V}$ is not restricted to be the first $K$ elements. We use the more conventional definition here, but our proposed algorithm is applicable to the general case.}.
A $K$-{bandlimited} graph signal can be expressed as
${\mathbf{x} = {\mathbf{V}_K}{{\tilde {\mathbf{x}}}_K}}$, where
${\mathbf{V}_K}$ means the first ${K}$ columns of ${\mathbf{V}}$, and ${{{\tilde {\mathbf{x}}}_K}}$ denotes the first ${K}$ elements of ${\tilde {\mathbf{x}}}$. 
To describe sampling on bandlimited graph signals, we define a sampling operator at first \cite{CMUicassp}.

\begin{definition}
{To sample $M$ elements from $\mathbf{x}$ to produce ${\mathbf{x}_{\mathcal{S}}=\mathbf{Cx} \in {{\rm{\mathbb{R}}}^M}}$ with ${\left| \mathcal{S} \right| = M}$ and ${\mathcal{S} \subseteq \mathcal{V}}$, we define a 0-1 binary \textit{sampling matrix} $\mathbf{C}$ as}
\begin{equation}\label{sampling}
\begin{split}
{{\mathbf{C}_{ij}} = \left\{ \begin{array}{l}
1,~~~{j = {\mathcal{S}(i)}};\\
0,~~~\textrm{otherwise},
\end{array} \right.}
\end{split}
\end{equation}
where $\mathcal{S}$ is the set of sampling indices, $\mathcal{S}(i)$ means the $i$-th element of set $\mathcal{S}$, and ${\left| \mathcal{S} \right|}$ is the number of elements in ${\mathcal{S}}$.
\end{definition}

A sampled $K$-bandlimited graph signal can now be written as ${{\mathbf{x}_\mathcal{S}}= \mathbf{C} {\mathbf{V}_K}{{\tilde {\mathbf{x}}}_K}}$.
Sampling operators satisfying $\textrm{rank}{{{\left( \mathbf{C}{\mathbf{V}_K} \right)}}} = K$ are called \emph{qualified sampling operators} in \cite{SCsampling}, since any $K$-bandlimited graph signal $\mathbf{x}$ can be perfectly recovered from samples $\mathbf{x}_\mathcal{S}$ on those operators.
Specifically, if $\textrm{rank}{{{\left( \mathbf{C}{\mathbf{V}_K} \right)}}} = K$, then, from a sampled graph signal, the LS solution can provide a perfect and unique reconstruction in noiseless environment, \textit{i.e.}, $\hat{\mathbf{x}}_{\text{LS}} = {\mathbf{V}_K}{\left( {\mathbf{C} {\mathbf{V}_K}} \right)^{\dagger}}\mathbf{x}_\mathcal{S}$, where $\dag$ denotes the pseudo-inverse computation \cite{AO}.
Further, $\textrm{rank}{{{\left( \mathbf{C}{\mathbf{V}_K} \right)}}} = K$ was demonstrated empirically to be fulfilled with high probability via random node selection when $M \geq K$ \cite{cornell}.
In the sequel, we focus on the sampling region $M \geq K$ and assume $\textrm{rank}{{{\left( \mathbf{C}{\mathbf{V}_K} \right)}}} = K$ is satisfied.

When the samples $\mathbf{x}_{\mathcal{S}}$ are corrupted by additive noise $\mathbf{n}_{\mathcal{S}}$, the LS recovery will produce a \textit{minimum variance unbiased} estimator $\hat{\mathbf{x}}_{\text{LS}} = {\mathbf{V}_K}{\left( {\mathbf{C} {\mathbf{V}_K}} \right)^{\dagger}}\mathbf{y}_\mathcal{S}={\mathbf{V}_K}{\left( {\mathbf{C} {\mathbf{V}_K}} \right)^{\dag}}(\mathbf{x}_\mathcal{S}+\mathbf{n}_{\mathcal{S}})=\mathbf{x}+{\mathbf{V}_K}{\left( {\mathbf{C} {\mathbf{V}_K}} \right)^{\dagger}}\mathbf{n}_{\mathcal{S}}$ \cite{Statistical signal}.
Assuming that noise $\mathbf{n}_{\mathcal{S}}$ is i.i.d. with zero mean and unit variance, the covariance matrix of the reconstruction error is
$\mathbf{R}_{\hat{\mathbf{x}}_{\text{LS}}}=\mathbb{E}\left[(\hat{\mathbf{x}}_{\text{LS}}-\mathbb{E}[\hat{\mathbf{x}}_{\text{LS}}])
(\hat{\mathbf{x}}_{\text{LS}}-\mathbb{E}[\hat{\mathbf{x}}_{\text{LS}}])^{\top}\right]
={\mathbf{V}_K}\left[ ({\mathbf{C}{\mathbf{V}_K})^{\top}}\mathbf{C} {\mathbf{V}_K} \right]^{ - 1}{\mathbf{V}^{\top}_K}$.
By the theory of optimal experiments design \cite{experiments}, finding a sampling matrix $\mathbf{C}$ to minimize the trace of the covariance matrix leads to the known \textit{A-optimality} criterion:
\begin{equation}\label{original formulation}
\begin{split}
\mathbf{C}^{*}=\mathop {\arg \min }\limits_{\mathbf{C} \in \mathbb{F}^{M\times N}}\textrm{tr}\left({[ {{{(\mathbf{C}{\mathbf{V}_K})}^{\top}}\mathbf{C}{\mathbf{V}_K}} ]^{ - 1}}\right).
\end{split}
\end{equation}
where $\mathbb{F}^{M\times N}$ the set of all $\mathbf{C}$ defined in \eqref{sampling}.

Minimizing the A-optimality objective \eqref{original formulation} directly using $\mathbf{C}$ is difficult because finding an optimal $\mathbf{C}^{*}$ is NP-hard and evaluating the A-optimality value for a given $\mathbf{C}$ requires expensive computations of $\mathbf{V}_K$ and matrix inverse. 

For simplicity, in the sequel the complement set of $\mathcal{S}$ is denoted by $\mathcal{S}^{c}$.
$|\mathcal{E}|$ means the cardinality of set $\mathcal{E}$, and ${{\mathbf{A}_{{\mathcal{S}_1}{\mathcal{S}_2}}}}$  is the sub-matrix of a matrix ${\mathbf{A}}$ with rows and columns indexed by ${\mathcal{S}_1}$ and ${\mathcal{S}_2}$, respectively.
${{\mathbf{A}_{{\mathcal{S}}{\mathcal{S}}}}}$ is simplified to ${{\mathbf{A}_{{\mathcal{S}}}}}$.
$\mathbf{I}$ is the identity matrix whose dimension depends on the context. 

%% file: aoptimal.tex
{In this section, we will propose a new sampling objective to approximate the A-optimality criterion. 
We first review the \emph{Neumann series theorem} in \cite{neumannseries} and one result in \cite{SPL}, both of which will be used for proving a forthcoming Theorem.}

\begin{theorem}
\label{neumann}
({\textbf{Neumann series theorem}}) If the eigenvalues $\lambda_i$ of a square matrix $\mathbf{A}$ have the property that $\rho(\mathbf{A}) \doteq \mathop {\max }\limits_i \left| {{\lambda _i}} \right| <1$, then its Neumann series $\mathbf{I} + \mathbf{A} + {\mathbf{A}^2} +  \cdots $ converges to ${(\mathbf{I} - \mathbf{A})^{ - 1}}$, \textit{i.e.}, ${(\mathbf{I} - \mathbf{A})^{ - 1}}=\sum\limits_{l=0}^\infty {\mathbf{A}^l}$.
\end{theorem}
\begin{proposition}
\label{infinite series}
Denote by $\mathbf{\Psi}  = {\left( {\mathbf{C}{\mathbf{V}_K}} \right)^{\top}}\mathbf{C}{\mathbf{V}_K}$, $\mathbf{\Phi}=\mathbf{I}-\mathbf{\Psi}$ and ${\delta _1} \le \ldots  \le {\delta _K}$ as the eigenvalues of $\mathbf{\Phi}$. Then,
\begin{equation}\label{eigenScope}
\begin{split}
0 \le {\delta _i}<1
\end{split}
\end{equation}
if ${{ \mathbf{C}{\mathbf{V}_K}}}$ is full column rank, \textit{i.e.},
$\emph{\textrm{rank}}{{{\left( \mathbf{C}{\mathbf{V}_K} \right)}}} = K$.
\end{proposition}
\begin{proof}
The proof is in Appendix \ref{eigenScope}.
\end{proof}
{These two results will support us to prove the Theorem \ref{shift proposition} in the next subsection.}

\subsection{Augmented A-optimality Criterion}

We propose an augmented objective by adding a weighted identity matrix to \eqref{original formulation}:
\begin{equation}\label{shift formulation}
\begin{split}
\mathbf{C}^{*}=\mathop {\arg \min }\limits_{\mathbf{C} \in \mathbb{F}^{M\times N}}{\textrm{tr}}\left({[ {{{(\mathbf{C}{\mathbf{V}_K})}^{\top}}\mathbf{C}{\mathbf{V}_K}}+\mu\mathbf{I}]^{ - 1}}\right)
\end{split}
\end{equation}
where $\mu$ is a small weight (shift) parameter with $0<\mu<1$.
We discuss its design in details later.

Problem \eqref{shift formulation} is also difficult to optimize directly like equation \eqref{original formulation}. Before proposing its alternative simpler form, we first present the next Proposition which will be employed to prove the following Theorem.
\begin{proposition}
\label{eigenvalue}
Denote by ${\epsilon _1} \le \ldots  \le {\epsilon_M}$ the eigenvalues of ${\mathbf{T}_\mathcal{S}}$, then
\begin{equation}
0\leq \epsilon _i\leq1
\end{equation}
where ${\mathbf{T}}={\mathbf{V}_K}{\mathbf{V}^{\top}_K}$ is an ideal LP graph filter with cutoff frequency $K$.
\end{proposition}
\begin{proof}
$\forall \mathbf{x} \in {\mathbb{R}^M}$ and $\left\| \mathbf{x} \right\|_2 = 1$,
$$ {\mathbf{x}^{\top}}\mathbf{T}_\mathcal{S} \mathbf{x} = {\left( \mathbf{C}^{\top}{\mathbf{x}} \right)^{\top}}\left( {{\mathbf{V}_K}{\mathbf{V}^{\top}_K}} \right)\left( \mathbf{C}^{\top}{\mathbf{x}} \right)={\mathbf{y}^{\top}}\left( {{\mathbf{V}_K}{\mathbf{V}^{\top}_K}} \right)\mathbf{y}
$$
where $\mathbf{y}=\mathbf{C}^{\top}{\mathbf{x}}$.
Hence $\mathbf{y}_\mathcal{S}={\mathbf{x}}$, $\mathbf{y}_{\mathcal{S}^{c}}={\mathbf{0}}$ and $\left\| \mathbf{y} \right\|_2 = 1$.

{Since $\mathbf{T}$ is symmetric and its eigen-decomposition can be written as $\mathbf{T}={{\mathbf{V}_K}{\mathbf{V}^{\top}_K}}={\mathbf{V}}\text{diag}\{1,...,1,0,...,0\}\mathbf{V}^{\top}$, its eigenvalues are 0 or 1.}
Then, due to the {Rayleigh quotient theorem} \cite{neumannseries}, $\forall \mathbf{y} \in {\mathbb{R}^N}$ and $\left\| \mathbf{y} \right\|_2 = 1$, $0 \le {\mathbf{y}^{\top}}({{\mathbf{V}_K}{\mathbf{V}^{\top}_K}})\mathbf{y} \leq 1$, which holds also for some specific $\mathbf{y}$.
Thus,
\begin{equation}\label{eigenTs2}
\begin{split}
0\leq {\mathbf{x}^{\top}}\mathbf{T}_\mathcal{S} \mathbf{x}\leq 1
 \end{split}
\end{equation}
which implies Proposition \ref{eigenvalue} via the Rayleigh quotient theorem.
\end{proof}

Now, we propose and prove the next Theorem by employing the results from Theorem \ref{neumann}, Proposition \ref{infinite series} and Proposition \ref{eigenvalue}.
\begin{theorem}\label{shift proposition}
The augmented A-optimal objective \eqref{shift formulation} is equivalent to
\begin{equation}\label{exchange formulation}
\begin{split}
{\mathcal{S}^*} = \mathop {\arg \min }_{\mathcal{S}:|\mathcal{S}|=M} \emph{\textrm{tr}}\left({\mathbf{T}_\mathcal{S}}+ \mu\mathbf{I} \right)^{-1},
 \end{split}
\end{equation}
if the selected matrix ${{ \mathbf{C}{\mathbf{V}_K}}}$ is full column rank.
\end{theorem}
\begin{proof}
Let $\tilde{\mathbf{\Psi}}= {{{(\mathbf{C}{\mathbf{V}_K})}^{\top}}\mathbf{C}{\mathbf{V}_K}}+\mu\mathbf{I}=\mathbf{\Psi}+\mu\mathbf{I}$, and $\tilde{\mathbf{\Phi}}=\mathbf{I}-\tilde{\mathbf{\Psi}}=\mathbf{\Phi}-\mu\mathbf{I}$.
{We know from Proposition 1 that $0 \le {\delta _i}<1$.}
Hence the eigenvalues of $\tilde{\mathbf{\Phi}}$ are in [$-\mu$,$1-\mu$).
Since the shift parameter $0<\mu< 1$, $\rho(\tilde{\mathbf{\Phi}})<1$, and {by applying Theorem \ref{neumann}}
\begin{equation}\label{shift Neumann series}
\begin{split}
[ {{{(\mathbf{C}{\mathbf{V}_K})}^{\top}}\mathbf{C}{\mathbf{V}_K}}+\mu\mathbf{I}]^{ - 1}=\sum_{l=0}^\infty(\mathbf{I}-\tilde{\mathbf{\Psi}})^l
\end{split}
\end{equation}
Leveraging on a property of the trace operator, we write
\begin{align}\label{trace property}
&\hspace{0.3cm}\textrm{tr}(\mathbf{I}-\tilde{\mathbf{\Psi}})^l
=\textrm{tr}\left[(1-\mu)\mathbf{I}-\mathbf{\Psi}\right]^l\nonumber\\
&\hspace{0cm}=\textrm{tr} \left[(1-\mu)^l\mathbf{I}+\sum\limits_{d = 1}^{l}\tbinom{l}{d}(1-\mu)^{(l-d)}(-\mathbf{\Psi})^{d}\right]\nonumber\\
&\hspace{0cm}\mathop =\limits^{(a)}(1-\mu)^lK+\sum\limits_{d = 1}^{l}\tbinom{l}{d}(1-\mu)^{(l-d)}\textrm{tr}\left[(-{\mathbf{T}_\mathcal{S}})^{d}\right]\\
&\hspace{0cm}=(1-\mu)^l(K-M)+\sum\limits_{d = 0}^{l}\tbinom{l}{d}(1-\mu)^{(l-d)}\textrm{tr}\left[(-{\mathbf{T}_\mathcal{S}})^{d}\right]\nonumber\\
&\hspace{0cm}=(1-\mu)^l(K-M)+\textrm{tr}\left[(1-\mu)\mathbf{I}-{\mathbf{T}_\mathcal{S}}\right]^l\nonumber
\end{align}
where $(a)$ holds because $\textrm{tr}(\mathbf{\Psi}^{d})=\textrm{tr}({\mathbf{V}^{\top}_K}{ {\mathbf{C}}^{\top}}\mathbf{C}{\mathbf{V}_K}...{\mathbf{V}^{\top}_K}{ {\mathbf{C}}^{\top}}\\
\mathbf{C}{\mathbf{V}_K})=\textrm{tr}(
\mathbf{C}{\mathbf{V}_K}{\mathbf{V}^{\top}_K}{{\mathbf{C}}^{\top}}...\mathbf{C}{\mathbf{V}_K}{\mathbf{V}^{\top}_K}{ {\mathbf{C}}^{\top}})
 =\textrm{tr}(\mathbf{T}^{d}_{\mathcal{S}})$.

Hence,
\begin{align}\label{infinite trace sum}
&\hspace{0.3cm}\textrm{tr}\left[\sum_{l=0}^\infty(\mathbf{I}-\tilde{\mathbf{\Psi}})^l\right]=
\sum_{l=0}^\infty\textrm{tr}(\mathbf{I}-\tilde{\mathbf{\Psi}})^l\nonumber\\
&\hspace{0cm}=\sum_{l=0}^\infty(1-\mu)^l(K-M)+\textrm{tr}\left[(1-\mu)\mathbf{I}-{\mathbf{T}_\mathcal{S}}\right]^l\\
&\hspace{0cm}=(K-M)\frac{1-(1-\mu)^{\infty}}{\mu}+
\textrm{tr}\sum_{l=0}^\infty\left[(1-\mu)\mathbf{I}-{\mathbf{T}_\mathcal{S}}\right]^l\nonumber\\
&\hspace{0cm}\mathop =\limits^{(a)}\frac{K-M}{\mu}+
\textrm{tr}\sum_{l=0}^\infty\left[\mathbf{I}-({\mathbf{T}_\mathcal{S}}+\mu\mathbf{I})\right]^l\nonumber
\end{align}
where equality $(a)$ holds since $0<\mu< 1$.

From Proposition \ref{eigenvalue}, the eigenvalues of $\mathbf{I}-({\mathbf{T}_\mathcal{S}}+\mu\mathbf{I})$ are in [$-\mu$,$1-\mu$].
Again, since the shift parameter $\mu$ is small and positive, $\rho\left[\mathbf{I}-({\mathbf{T}_\mathcal{S}}+\mu\mathbf{I})\right]\leq1-\mu<1$.
{Using Theorem \ref{neumann}}, we have
\begin{equation}\label{back transformation}
\begin{split}
\sum_{l=0}^\infty\left[\mathbf{I}-({\mathbf{T}_\mathcal{S}}+\mu\mathbf{I})\right]^l=
({\mathbf{T}_\mathcal{S}}+\mu\mathbf{I})^{-1}
\end{split}
\end{equation}

Combining \eqref{shift Neumann series}, \eqref{infinite trace sum} and \eqref{back transformation},
\begin{align}\label{new proposition}
\textrm{tr}\left({[ {{{(\mathbf{C}{\mathbf{V}_K})}^{\top}}\mathbf{C}{\mathbf{V}_K}}+\mu\mathbf{I}]^{ - 1}}\right)=\frac{K-M}{\mu}+\textrm{tr}\left({\mathbf{T}_\mathcal{S}}+ \mu\mathbf{I} \right)^{-1}\nonumber
\end{align}
which implies Theorem \ref{shift proposition} because $M$, $K$ and $\mu$ are all constant in a sampling problem.
\end{proof}

{Optimizing the new problem \eqref{exchange formulation} requires computation of an ideal LP filter rather than first $K$ eigenvectors of $\mathbf{L}$.
We discuss an efficient approximation of the LP filter next. }

\subsection{Ideal Low-pass Graph Filter Approximation}

One classical method is to approximate an ideal LP filter $\mathbf{T}$ in \eqref{exchange formulation} by using a Chebyshev matrix polynomial of $\mathbf{L}$, \textit{i.e.}, $\mathbf{T}^{\textrm{Poly}}=\sum _{i=1}^{n}\left( \sum_{j=0}^q \beta_j\lambda^j_i\right)\mathbf{v}_i\mathbf{v}^{\top}_i=\sum_{j=0}^q \beta_j \mathbf{L}^j$ \cite{wavelet}.
Another method is {to apply the Jacobi eigenvalue algorithm}: approximately diagonalize $\mathbf{L}$ with an estimated eigenvector matrix $\tilde{\mathbf{V}}=\mathbf{S}_1...\mathbf{S}_J$ \cite{France Givens}.
Specifically, the goal is to solve the following optimization problem:
\begin{equation}\label{GivensOpt}
\mathop {\min}_{\hat{\pmb{\Lambda}},\mathbf{S}_1,...,\mathbf{S}_J}~~~\|\mathbf{L}-\mathbf{S}_1...\mathbf{S}_J
\hat{\pmb{\Lambda}}\mathbf{S}_J^{\top}...\mathbf{S}_1^{\top}\|^2_F
\end{equation}
where $\hat{\pmb{\Lambda}}$ is constrained to be a diagonal matrix {and $\mathbf{S}_i$ are constrained to be Givens rotation matrices.}

A truncated Jacobi algorithm \cite{Jacobi} was adopted to optimize \eqref{GivensOpt} iteratively.
%
{With the optimized $\tilde{\mathbf{V}}$ where $\mathbf{L}=\tilde{\mathbf{V}}\hat{\pmb{\Lambda}}\tilde{\mathbf{V}}^{\top}$},
an ideal LP filte $\mathbf{T}$ can be implemented as $\mathbf{T}^{\textrm{FGFT}}=\tilde{\mathbf{V}}_K\tilde{\mathbf{V}}^{\top}_K$.

We adopt $\mathbf{T}^{\textrm{FGFT}}$ for the following reasons:
\begin{enumerate}[(1)]
  \item \cite{France Givens} claimed the superiority of $\mathbf{T}^{\textrm{FGFT}}$ compared with $\mathbf{T}^{\textrm{Poly}}$ when approximating ideal graph LP filters; specifically, $\frac{\left\|\mathbf{T}-\mathbf{T}^{\textrm{FGFT}}\right\|_F}{\left\|\mathbf{T}\right\|_F} < \frac{\left\|\mathbf{T}-\mathbf{T}^{\textrm{Poly}}\right\|_F}{\left\|\mathbf{T}\right\|_F}$ under the same approximation complexity \cite{France Givens}.
  In our work, $\mathbf{T}$ is an ideal graph LP filter. Moreover, if the bandwidth of graph signals $K$ varies in one graph, we can reuse the same $\tilde{\mathbf{V}}$ when realizing $\mathbf{T}^{\textrm{FGFT}}$.
  \item  A sufficient condition of Proposition \ref{eigenvalue} is that the eigenvalues of $\mathbf{T}$ are in $\{0,1\}$, based on which we can claim the eigenvalues of $\mathbf{I}-({\mathbf{T}_\mathcal{S}}+\mu\mathbf{I})$ are in [$-\mu$,$1-\mu$] in Theorem 2.
  Then, \eqref{back transformation} can hold with a small positive parameter $\mu$.
  Small $\mu$ is also required for \eqref{shift formulation} to approximate the original A-optimality criterion \eqref{original formulation} well.
  Because $\tilde{\mathbf{V}}$ estimated from Givens matrices is orthogonal, Proposition \ref{eigenvalue} also holds for $\mathbf{T}^{\textrm{FGFT}}_{\mathcal{S}}$.
  In contrast, the eigenvalue scope of $\mathbf{T}^{\textrm{Poly}}_{\mathcal{S}}$ is unpredictable.
\end{enumerate}

We now write our final augmented objective as
\begin{equation}\label{final formula}
\begin{split}
{\mathcal{S}^*} = \mathop {\arg \min }_{\mathcal{S}:|\mathcal{S}|=M} \textrm{tr}\left({\mathbf{T}^{\textrm{FGFT}}_\mathcal{S}}+ \mu\mathbf{I} \right)^{-1}
 \end{split}
\end{equation}
which requires just an approximated LP filter operator rather than eigen-decomposition of the graph Laplacian operator.
For simplicity, we write $\mathbf{G}=\mathbf{T}^{\textrm{FGFT}}+ \mu\mathbf{I}$ in the sequel.

\subsection{Selection of Shift Parameter $\mu$}
We now discuss how to select an appropriate shift $\mu$.
To well approximate the original criterion \eqref{original formulation}, the shift $\mu$ ought to be as small as possible, but a small $\mu$ would cause the matrix inverse computation in \eqref{final formula} to be unstable.
To ensure numerical stability, we can bound the \textit{condition number} $\kappa(\mathbf{G})$ of $\mathbf{G}$, where $\kappa(\mathbf{G})=\lambda_{\text{max}}(\mathbf{G})/\lambda_{\text{min}}(\mathbf{G})$, using $\mu$ as follows \cite{condition number}.
As discussed, Proposition \ref{eigenvalue} also holds for $\mathbf{T}^{\textrm{FGFT}}_{\mathcal{S}}$, so $\mu\leq\lambda\left(\mathbf{G}_{\mathcal{S}}\right)\leq1+\mu$.
Hence we can bound $\kappa(\mathbf{G}_{\mathcal{S}})$ as follows:
\begin{equation}\label{CNbound}
    \kappa(\mathbf{G}_{\mathcal{S}})=\frac{\lambda_{\text{max}}(\mathbf{G}_{\mathcal{S}})}
    {\lambda_{\text{min}}(\mathbf{G}_{\mathcal{S}})}\leq\frac{1+\mu}{\mu}\leq \kappa_0
\end{equation}
where $\kappa_0$ is the upper bound of an acceptable condition number.
From the latter part, $\mu\geq\frac{1}{\kappa_0-1}$.
Because $\mu$ should be as small as possible, the optimal $\mu$ is $\mu^{*}=\frac{1}{\kappa_0-1}$.
We will set $\kappa_0=100$ in our experiments and compare its performance with other design strategies via simulations. 

%% file: greedy.tex
It is difficult to compute an optimal solution of our proposed criterion \eqref{final formula} with reasonable complexity since the problem is combinatorial.
Greedy algorithms have been commonly used in the graph sampling literature \cite{AO,SCsampling}.
However, if \eqref{final formula} is minimized na\"{i}vely using a greedy scheme, then the algorithm needs to perform one matrix inversion to evaluate each candidate sample set $\mathcal{S}$.
To alleviate this computation burden, we propose an accelerated greedy algorithm to avoid the matrix inverse operation.

\subsection{Accelerated Greedy Sampling}

Our accelerated algorithm is based on the \emph{block matrix inversion  formula}, reviewed as follows \cite{blockwise}.
\begin{lemma}\label{matrix Lemma1}
The inverse of matrix $\mathbf{M}$ can be computed using the inverse of sub-matrix $\mathbf{A}$ and the inverse of its Schur complement, \textit{i.e.},
\begin{align}\label{inversion lemma1}
&\hspace{0.3cm}
 \mathbf{M}^{-1} = \left[ \begin{array}{l}
 \mathbf{A}~~~{\mathbf{U}} \\
 {\mathbf{V}}~~~{\mathbf{C}} \\
 \end{array} \right]^{-1}\\
 &\hspace{0cm}=\left[ \begin{array}{l}
 \mathbf{A}^{-1}+\mathbf{A^{-1}UH^{-1}VA^{-1}}~~~-{\mathbf{A^{-1}UH^{-1}}} \nonumber\\
 -{\mathbf{H^{-1}VA^{-1}}}~~~~~~~~~~~~~~~~~~~~~~~~~{\kern 1pt}{\mathbf{H^{-1}}} \\
 \end{array} \right]
\end{align}
where $\mathbf{H}=\mathbf{C-VA^{-1}U}$ is the Schur complement $\mathbf{M} / \mathbf{A}$ of sub-matrix $\mathbf{A}$ of matrix $\mathbf{M}$.
\end{lemma}

We observe that matrix $\mathbf{G}$ is symmetric, and under some permutation, its sub-matrix $\mathbf{G}_{\mathcal{S}\cup\{i\}}$ can be expressed as
\begin{align}\label{inversion lemma1}
\mathbf{G}_{\mathcal{S}\cup\{i\}} = \left[ \begin{array}{l}
 \mathbf{G}_\mathcal{S}~~~~~~~{\mathbf{G}}_{\mathcal{S},\{i\}} \\
{\mathbf{G}}_{\{i\},\mathcal{S}}~~~{\mathbf{G}}_{ii} \\
 \end{array} \right]\doteq \left[ \begin{array}{l}
 \mathbf{G}_\mathcal{S}~~~\mathbf{g}_i \\
~\mathbf{g}^{\top}_i~~\mathbf{G}_{ii} \\
 \end{array} \right]
\end{align}
where $\mathbf{g}_i$ denotes the partial vector of $i$-th column of $\mathbf{G}$ indexed by $\mathcal{S}$.

From Lemma \ref{matrix Lemma1}, we can compute the inverse of $\mathbf{G}_{\mathcal{S}\cup\{i\}}$ as
\begin{align}\label{inversion lemma1}
\mathbf{G}^{-1}_{\mathcal{S}\cup\{i\}} = \left[ \begin{array}{l}
\mathbf{G}^{-1}_{\mathcal{S}}+h^{-1}\mathbf{G}^{-1}_{\mathcal{S}}\mathbf{g}_i\mathbf{g}^{\top}_i\mathbf{G}^{-1}_{\mathcal{S}}
~~~~-h^{-1}\mathbf{G}^{-1}_{\mathcal{S}}\mathbf{g}_i \\
~~~~~~~~~-h^{-1}\mathbf{g}^{\top}_i\mathbf{G}^{-1}_{\mathcal{S}}~~~~~~~~~~~~~~~h^{-1} \\
 \end{array} \right]
\end{align}
where $h=\mathbf{G}_{ii}-\mathbf{g}^{\top}_i\mathbf{G}^{-1}_{\mathcal{S}}\mathbf{g}_i$ is actually a scalar here.

$\mathbf{G}^{-1}_{\mathcal{S}\cup \{i\}}$ can thus be computed using $\mathbf{G}^{-1}_{\mathcal{S}}$ stored in the last iteration.
Instead of computing $\mathbf{G}^{-1}_{\mathcal{S}\cup \{i\}}$ directly, \eqref{inversion lemma1} performs {only two} matrix-vector products and then stacks the results.
We outline our algorithm in pseudo-code in {Algorithm \ref{GFS}} and call it \emph{graph filter submatrix} (GFS)-based sampling algorithm.
We note that GFS requires neither computation of eigenvectors of $\mathbf{L}$ nor large matrix inversion, and it produces exactly the same sampling set as na\"{i}vely optimizing problem \eqref{final formula} with a greedy scheme.

We next study the performance bound of this greedy solution via super-modularity analysis.


\subsection{Super-modularity Analysis}

Let $f(\mathcal{S})=\textrm{tr}\left({\mathbf{T}^{\textrm{FGFT}}_\mathcal{S}}+ \mu\mathbf{I} \right)^{-1}$ and
$g(\mathcal{S})=\text{tr}[(\mathbf{C}\tilde{\mathbf{V}}_K)^{\top}\mathbf{C}\tilde{\mathbf{V}}_K+\mu\mathbf{I}]
^{-1}$.
Because Proposition \ref{eigenvalue} holds for $\mathbf{T}^{\textrm{FGFT}}_{\mathcal{S}}$, following a similar proof to one in Theorem \ref{shift proposition}, we have $f(\mathcal{S})=g(\mathcal{S})+\frac{M-K}{\mu}$.
Hence, the solution of greedily optimizing $f(\mathcal{S})$ is exactly the same as that of greedily optimizing $g(\mathcal{S})$, but with much lower complexity.
Since $g(\mathcal{S})$ approximates the real MSE value, it is reasonable to evaluate the sub-optimality of its greedy solution.
\begin{algorithm}[tp]
\caption{GFS graph signal sampling algorithm}
\label{GFS}
\textbf{Input:} Graph operator $\mathbf{L}$, {bandwidth $K$}, budget ${M}$ and parameter $\mu$. $\mathcal{S}=\{\varnothing\}$.
\begin{algorithmic}[1]
\STATE Compute $\tilde{\mathbf{V}}=\mathbf{S}_1...\mathbf{S}_J$ of $\mathbf{L}$ via \eqref{GivensOpt}.
\STATE Compute $\mathbf{T}^{\textrm{FGFT}}=\tilde{\mathbf{V}}_K\tilde{\mathbf{V}}^{\top}_K$ and $\mathbf{G}=\mathbf{T}^{\textrm{FGFT}}+\mu\mathbf{I}$.
\STATE Select the first node by $u = \mathop {\arg \max }\limits_{i \in \mathcal{V}} \mathbf{G}_{ii} $, update $\mathbf{G}^{-1}_{\mathcal{S}}={\mathbf{G}^{-1}_{uu}}$ and ${\mathcal{S} \leftarrow \mathcal{S} \cup \left\{ u \right\}}$
\STATE\textbf{While} $\left| \mathcal{S} \right| < M$
\STATE$\forall i \in \mathcal{S}^{c}$, compute\\
$\mathbf{g}_i={\mathbf{G}}_{\mathcal{S},\{i\}}$ and $h=\mathbf{G}_{ii}-\mathbf{g}^{\top}_i\mathbf{G}^{-1}_{\mathcal{S}}\mathbf{g}_i$\\
$\mathbf{G}^{-1}_{\mathcal{S}\cup\{i\}} = \left[ \begin{array}{l}
\mathbf{G}^{-1}_{\mathcal{S}}+h^{-1}\mathbf{G}^{-1}_{\mathcal{S}}\mathbf{g}_i\mathbf{g}^{\top}_i\mathbf{G}^{-1}_{\mathcal{S}}
~-h^{-1}\mathbf{G}^{-1}_{\mathcal{S}}\mathbf{g}_i \\
~~~~~~-h^{-1}\mathbf{g}^{\top}_i\mathbf{G}^{-1}_{\mathcal{S}}~~~~~~~~~~~~~~~h^{-1} \\
 \end{array}\right]$\\
\STATE Select $u = \mathop {\arg \min }\limits_{i \in \mathcal{S}^{c}} \textrm{tr}\left[\mathbf{G}^{-1}_{\mathcal{S}\cup\{i\}}\right] $\\
\STATE Update $\mathbf{G}^{-1}_{\mathcal{S}}=\mathbf{G}^{-1}_{\mathcal{S}\cup\{u\}}$ and ${\mathcal{S} \leftarrow \mathcal{S} \cup \left\{ u \right\}}$\\
\STATE \textbf{end While}\\
\STATE Return $\mathcal{S}$ and $\mathbf{G}^{-1}_{\mathcal{S}}$
\end{algorithmic}
\end{algorithm}

\begin{table*}
\caption{Complexity Comparison of Different Sampling Strategies}
\label{complexity}
\begin{center}
\begin{tabular}{ccc}
 \hline
\hline
  \textbf{}&{Preparation}&{Selection}\\\hline
Spectral Proxies &\emph{NONE}&$\mathcal{O}\left({k\left| \mathcal{E} \right|M{T_2}\left( k \right)}+NM\right)$ \\
E-optimal &$\mathcal{O}\left( {\left( {\left| \mathcal{E} \right|M + R{M^3}} \right){T_1}} \right)$&$\mathcal{O}\left(NM^{4}\right)$\\
{MFN} &{$\mathcal{O}\left( {\left( {\left| \mathcal{E} \right|M + R{M^3}} \right){T_1}} \right)$}&{$\mathcal{O}\left(NM^{4}\right)$}\\
MIA&$\mathcal{O}(qN|\mathcal{E}|)$&$\mathcal{O}\left(NLM^{3.373}\right)$\\
Proposed GFS&$\mathcal{O}(N^{2}\text{log}^{2}N)$&$\mathcal{O}\left(NM^{3}\right)$\\\hline\hline
 \end{tabular}
 \end{center}
\end{table*}

\begin{lemma}\label{submodular}
 The set function $g(\mathcal{S})=\emph{\text{tr}}[(\mathbf{C}\tilde{\mathbf{V}}_K)^{\top}\mathbf{C}\tilde{\mathbf{V}}_K+\mu\mathbf{I}]
^{-1}$ is (i) monotone decreasing and (ii) $\alpha$-supermodular with
 \begin{equation}\label{alpha}
    \alpha \geq \frac{\mu^3(\mu+2)}{(\mu+1)^4}
 \end{equation}
\end{lemma}
\begin{proof}
The detailed proof is shown in Appendix \ref{submodularProof}
\end{proof}

Let $\mathcal{S}^{*}$ be the optimal solution of minimizing $g(\mathcal{S})$ when $|\mathcal{S}|=M$, and $\mathcal{S}$ be the set obtained by applying the GFS algorithm.
{We have the following result using} Lemma \ref{submodular}:
\begin{equation}\label{performanceBound}
    \frac{g(\mathcal{S})-g(\mathcal{S}^{*})}{g(\{\})-g(\mathcal{S}^{*})}\leq e^{-\alpha}
\end{equation}
which is easily derived from Theorem 2 in \cite{greedybound}.

{From \eqref{performanceBound}, the value of $g(\mathcal{S})$ will be bounded by a function of $g(\mathcal{S^{*}})$, thus the sub-optimality of our GFS solution is up-bounded.  }

\subsection{Complexity Analysis}

We here analyze the complexity of the proposed GFS sampling strategy, assuming $M=K$.
The overall complexity is divided into two parts: ``Preparation'' and ``Selection''.
The preparation step includes the complexity of computing the prior information, \textit{e.g.}, eigen-decomposition, and the selection step has the complexity of collecting samples.
GFS needs to compute matrix $\mathbf{G}$ in the preparation step.
Authors in \cite{France Givens} has shown that the complexity of computing $\tilde{\mathbf{V}}$ via parallel truncated Jacobi algorithm is $\mathcal{O}(NJ\text{log}N)$.
During experiments, the number of Givens rotation matrices is set at $J=\mathcal{O}(N\text{log}N)$.
Once we obtain $\mathbf{S}_1,...,\mathbf{S}_J$, $\mathbf{T}^{\textrm{FGFT}}$ can be implemented via $\mathbf{T}^{\textrm{FGFT}}=\mathbf{S}_1...\mathbf{S}_J \mathbf{B}\mathbf{S}_J^{\top}...\mathbf{S}_1^{\top}$, where $\mathbf{B}$ is a diagonal matrix with ones on the first $K$ diagonal elements and zeros for the rest.
Since every Givens matrix $\mathbf{S}_i$ is sparse with 4 non-zero entries, $\mathbf{T}^{\textrm{FGFT}}$ can be computed iteratively with complexity $\mathcal{O}(NJ)$.
For the selection step, the complexity for matrix-vector product is $\mathcal{O}(M^2)$.
Considering $|\mathcal{S}|=M$ and $|\mathcal{S}^c|=N-M$, the complexity of GFS in the sampling step is at most $\mathcal{O}\left(NM^{3}\right)$ as shown in Table\;\ref{complexity}.

Table\;\ref{complexity} shows the complexity of different sampling algorithms.
In the preparation step, the E-optimal and MFN sampling methods require $\mathbf{V}_K$, whose complexity is $\mathcal{O}\left( {\left( {\left| \mathcal{E} \right|M + R{M^3}} \right){T_1}} \right)$ via the block version of the Rayleigh quotient minimization method \cite{LOBPCG}.
$T_1$ denotes the average number of iterations for convergence of this method, and $R$ is a constant for mixing the complexity of two iterative steps in this method \cite{LOBPCG}.
MIA requires the $K$-th eigenvalue of $\mathbf{L}$ and an approximate LP filter  $\mathbf{T}^{\text{Poly}}$, whose complexity is $\mathcal{O}(qN|\mathcal{E}|)$, where $q$ is the length of a Chebyshev matrix polynomial.
In the selection step, the first eigen-pair of $\left((\mathbf{L}^{\top})^k\mathbf{L}^k\right)_{\mathcal{S}^c}$ is required for {SP}, and its complexity is $\mathcal{O}\left( {k\left| \mathcal{E} \right|M{T_2}\left( k \right)} \right)$, where $k$ is the order of the proxies approximation, and $T_2(k)$ is the number of iterations for convergence of one single eigen-pair \cite{AO}.
E-optimal requires SVD, and MFN requires matrix inversion computation for selection, so their complexity in this step are both $\mathcal{O}(NM^4)$.
The matrix multiplication in MIA has an asymptotic complexity of $\mathcal{O}(M^{2.373})$, thus its complexity in the selection step is $\mathcal{O}(NLM^{3.373})$, where $L$ is a truncation parameter \cite{SPL}.
{From Table \ref{complexity}, we can see that the complexity of the proposed GFS method is lower than the MIA, MFN and E-optimal method in the selection step. Numerical analysis for complexity will be illustrated in the experimental section.} 

%% file: NodeExchange.tex
Most existing graph sampling works assume a static graph from which to select nodes \cite{Uncertainty,SCsampling,AO}, although, as argued in the Introduction, the availability of sensing nodes varies over time for many practical applications.
In this section, we extend our proposed algorithm to address sampling on dynamic graphs: how to select nodes for sampling if the available nodes are a time-varying subset?
For concreteness, we first develop a mathematical model to generate a node subset $\mathcal{A}_t$ at time $t$ that slowly evolves over time.

{At $t=0$, we first randomly generate $P_0 N$ nodes to form an initial available node subset $\mathcal{A}_0$ from $\mathcal{V}$, where $P_0$ is a fixed probability.
Denote by $z^t_j$ the state of node $j$ at time $t$, where $z^t_j=1$ if $j\in\mathcal{A}_t$, and $z^t_j=0$ otherwise, \textit{i.e.} $j$ is in the complement set $\mathcal{A}^c_t$.
Define the state crossover probability as $P(z^{t+1}_j=0|z^{t}_j=1)=P(z^{t+1}_j=1|z^{t}_j=0)=\epsilon$.
Given $\mathcal{A}_t$, one can generate $\mathcal{A}_{t+1}$ from $\mathcal{A}_{t}$ using the state crossover probability.}
Note that the crossover probability $\epsilon$ tends to be small in practice.
The \textit{dynamic graph sampling} problem\footnote{If $\mathcal{S}_t$ at time $t$ is still available at $t+1$, we make no changes. While it is possible to swap newly available nodes at time $t+1$ with selected nodes in $\mathcal{S}_t$, we found the performance-complexity tradeoff not worthwhile. } can be stated formally as follows:
Given a sampling set $\mathcal{S}_t$ at time $t$, where $\mathcal{S}_t \subset \mathcal{A}_{t}$, how to compute $\mathcal{S}_{t+1}$ if $\mathcal{S}_t \not\subset \mathcal{A}_{t+1}$?

We next discuss our sampling strategy for this problem.


\subsection{Sampling on Dynamic Subsets}

For conciseness, we denote $\mathcal{P}=\mathcal{A}_{t+1}\cap\mathcal{S}_{t}$ to be the \textit{selected and available} (SA) set and $\mathcal{U}=\mathcal{S}_{t}\setminus\mathcal{P}$ as the \textit{selected and unavailable} (SU) node set.
Then, our sampling strategy based on node exchange (NE) can be described as follows:
\begin{enumerate}
  \item We replace a node $j\in\mathcal{U}$ with an \textit{available and unselected} (UA) node $k\in\mathcal{Q}$ where $\mathcal{Q}=\mathcal{A}_{t+1}\setminus\mathcal{P}$ is the UA node set.
  The optimal node $k^*$ is greedily chosen if exchanging $j$ with $k^*$ minimizes the objective function in all the $(j,k)$ exchange pairs.
  Repeating this procedure, all unavailable nodes in SU set are replaced by available nodes, which will form a new sampling set $\mathcal{\hat{S}}$.
  \item To further improve performance, node $p\in\mathcal{\hat{S}}$ will be replaced by node $q\in \mathcal{H}$ if the $(p,q)$-pair exchange can induce a lower objective, where $\mathcal{H}=\mathcal{A}_{t+1}\setminus\mathcal{\hat{S}}$.
  \item Step 2) will return the final sampling set $\mathcal{S}_{t+1}$ until it achieves the iteration number constraint, or it could not find one more node-pair exchange to make the objective decrease in this constraint.
\end{enumerate}

After each node-pair exchange, $\mathbf{G}^{-1}_{\mathcal{S}}$ has to be computed.
In the following, we propose one method to relieve this computation burden, in which the node-pair is generally indexed by $(j,k)$.
\vspace{-0.1in}
\subsection{One-pair Node Exchange}
The complexity of computing $\mathbf{G}^{-1}_{\mathcal{S}}$ after each node-pair exchange  can be reduced based on \emph{Sherman-Morrison formula} \cite{matrix lemma}, reviewed below.

\begin{lemma}\label{matrix lemma2} Suppose $\mathbf{A} \in \mathbb{R}^{N \times N}$ is an invertible square matrix and $\mathbf{u},\mathbf{v} \in \mathbb{R}^{N}$ are column vectors. If $\mathbf{A}+\mathbf{u}\mathbf{v}^{\top}$ is invertible, then its inverse is given by
\begin{equation}\label{sherman-morisson}
\begin{split}
\left(\mathbf{A}+\mathbf{u}\mathbf{v}^{\top}\right)^{-1}
=\mathbf{A}^{-1}-\frac{\mathbf{A}^{-1}\mathbf{u}\mathbf{v}^{\top}\mathbf{A}^{-1}}
{1+\mathbf{v}^{\top}\mathbf{A}^{-1}\mathbf{u}}
 \end{split}
\end{equation}
\end{lemma}
After exchanging node-pair $(j,k)$, we will have a new sample set $\mathcal{\tilde{S}}=\mathcal{S} \cup \left\{k\right\} \setminus \left\{j\right\}$.

It is observed that under some permutation
\begin{align}\label{original set}
\mathbf{G}_{\mathcal{S}} = \left[ \begin{array}{l}
{\mathbf{G}}_{\mathcal{S}\setminus\{j\}}~~~~~~~{\mathbf{G}}_{\mathcal{S}\setminus\{j\},\{j\}} \\
{\mathbf{G}}_{\{j\},\mathcal{S}\setminus\{j\}}~~~{\mathbf{G}}_{jj} \\
 \end{array} \right]
\end{align}
and
\begin{align}\label{update set}
\mathbf{G}_{\mathcal{\tilde{S}}} = \left[ \begin{array}{l}
{\mathbf{G}}_{\mathcal{S}\setminus\{j\}}~~~~~~~{\mathbf{G}}_{\mathcal{S}\setminus\{j\},\{k\}} \\
{\mathbf{G}}_{\{k\},\mathcal{S}\setminus\{j\}}~~~{\mathbf{G}}_{kk} \\
 \end{array} \right]
\end{align}

$\mathbf{G}_{\mathcal{\tilde{S}}}$ and $\mathbf{G}_{\mathcal{S}}$ share the same sub-matrix of dimension $(M-1)\times(M-1)$.
Hence, we can construct $\mathbf{G}_{\mathcal{\tilde{S}}}$ by modifying only the column and row in $\mathbf{G}_{\mathcal{S}}$ corresponding to node $j$.
Assuming that node $j$ is the $i$-th node in $\mathcal{S}$, \textit{i.e}, $j=\mathcal{S}(i)$, we first replace the $i$-th row in $\mathbf{G}_{\mathcal{S}}$ with $\mathbf{G}^{\top}_{\mathcal{\tilde{S}},\{k\}}$ to form an intermediate matrix $\mathbf{F}$, and then replace the $i$-th column of $\mathbf{F}$ with the column vector $\mathbf{G}_{\mathcal{\tilde{S}},\{k\}}$ to construct $\mathbf{G}_{\mathcal{\tilde{S}}}$.
Therefore, we can express $\mathbf{G}_{\tilde{\mathcal{S}}}$ as $\mathbf{G}_{\mathcal{S}}$ with two rank-1 updates:
\begin{align}\label{rank 1 update}
\mathbf{G}_{\mathcal{\tilde{S}}}=\mathbf{G}_{\mathcal{S}}+\mathbf{e}_i\mathbf{p}^{\top}
+\mathbf{q}\mathbf{e}^{\top}_i
\end{align}
where $\mathbf{p}=\mathbf{G}_{{\mathcal{\tilde{S}}},\{k\}}-\mathbf{G}_{\mathcal{S},\{j\}}$,
$\mathbf{q}=\mathbf{p}-(\mathbf{G}_{kk}-\mathbf{G}_{jj})\mathbf{e}_i$ and $\mathbf{e}_i$ is an indicator vector
with $\mathbf{e}_i(i)=1$ and $\mathbf{e}_i(j)=0$ for $j\neq i$.
As introduced, $\mathbf{F}=\mathbf{G}_{\mathcal{S}}+\mathbf{e}_i\mathbf{p}^{\top}$.

From Lemma \ref{matrix lemma2}, the inverse of $\mathbf{G}_{\mathcal{\tilde{S}}}$ can be computed via:
\begin{align}\label{rank 1 inverse1}
\mathbf{F}^{-1}=\mathbf{G}^{-1}_{\mathcal{S}}
-\frac{\mathbf{G}^{-1}_{\mathcal{S}}\mathbf{e}_i\mathbf{p}^{\top}\mathbf{G}^{-1}_{\mathcal{S}}}
{1+\mathbf{p}^{\top}\mathbf{G}^{-1}_{\mathcal{S}}\mathbf{e}_i}
\end{align}
\begin{align}\label{rank 1 inverse2}
\mathbf{G}^{-1}_{\mathcal{\tilde{S}}}=\mathbf{F}^{-1}
-\frac{\mathbf{F}^{-1}\mathbf{q}\mathbf{e}^{\top}_i\mathbf{F}^{-1}}
{1+\mathbf{e}^{\top}_i\mathbf{F}^{-1}\mathbf{q}}
\end{align}
composed of simple matrix-vector products only.

We compute the initial set $\mathcal{S}_0$ using our proposed GFS algorithm described in Section \ref{sec:greedy}.
The details of the dynamic subset sampling strategy are outlined in Algorithm \ref{GFSNE}, called GFS-NE sampling.
$K_0$ is the upper bound of exchange iterations to contain complexity.
In simulations, we will compare this proposed dynamic subset sampling algorithm to existing greedy sampling algorithms which select nodes from scratch at each time $t$ from time-varying subsets. \begin{algorithm}[tp]
\caption{GFS-NE dynamic subset sampling algorithm}
\label{GFSNE}
Apply the GFS algorithm on $\mathcal{A}_0$ and then output $\mathcal{S}_0$ and $\mathbf{G}^{-1}_{\mathcal{S}_0}$.
Given $K_0$, $\mathcal{S}_t$ and $\mathbf{G}^{-1}_{\mathcal{S}_t}$,
\begin{algorithmic}[1]
\STATE Initialization: $\xi=0$, $\mathcal{S}=\mathcal{S}_t$ and $\mathbf{G}^{-1}_{\mathcal{S}}=\mathbf{G}^{-1}_{\mathcal{S}_t}$\\
\STATE \textbf{For} $r=1,...,|\mathcal{U}|$\\
\STATE$j=\mathcal{U}(r)$ and $i=d|_{\mathcal{S}(d)=j}$\\
\STATE$\forall k \in \mathcal{Q}$, $\mathcal{\tilde{S}}=\mathcal{S} \cup\left\{k\right\} \setminus \left\{j\right\}$, compute $\mathbf{G}^{-1}_{\mathcal{\tilde{S}}}$ based on equations \eqref{rank 1 inverse1} and \eqref{rank 1 inverse2}.
\STATE Select $k^* = \mathop {\arg \min }\limits_{k \in \mathcal{Q}} \textrm{tr}\left[\mathbf{G}^{-1}_{\mathcal{\tilde{S}}}\right]$, and $\mathcal{\hat{S}}=\mathcal{S} \cup\left\{k^*\right\} \setminus \left\{j\right\}$\\
\STATE $\mathcal{S}=\mathcal{\hat{S}}$, $\mathbf{G}^{-1}_{\mathcal{S}}=\mathbf{G}^{-1}_{\mathcal{\hat{S}}}$ and $\mathcal{Q}=\mathcal{Q}\setminus\{k^*\}$
\STATE\textbf{end~For}\\
\end{algorithmic}
\begin{algorithmic}[1]
\STATE \textbf{For} $i=1,...,|\mathcal{S}|$\\
\STATE$j=\mathcal{S}(i)$ \\
\STATE\textbf{For} $r=1,...,|\mathcal{H}|$\\
\STATE$k=\mathcal{H}(r)$ and $\mathcal{\tilde{S}}=\mathcal{S} \cup\left\{k\right\} \setminus \left\{j\right\}$, compute $\mathbf{G}^{-1}_{\mathcal{\tilde{S}}}$ based on equations \eqref{rank 1 inverse1} and \eqref{rank 1 inverse2}\\
\STATE\textbf{If}~$\textrm{tr}(\mathbf{G}^{-1}_{\mathcal{\tilde{S}}})
          <\textrm{tr}(\mathbf{G}^{-1}_{\mathcal{S}})$\\
\STATE$\xi=\xi+1$, ${\mathcal{S}}={\mathcal{\tilde{S}}}$ and $\mathbf{G}^{-1}_{\mathcal{S}}=\mathbf{G}^{-1}_{\mathcal{\tilde{S}}}$\\
\STATE\textbf{If} $\xi\geq K_0$, return $\mathcal{S}_{t+1}=\mathcal{S}$ and $\mathbf{G}^{-1}_{\mathcal{S}_{t+1}}=\mathbf{G}^{-1}_{\mathcal{S}}$\\
\STATE\textbf{end If}\\
\STATE\textbf{end~For}\\
\STATE\textbf{end~For}\\
\STATE Return $\mathcal{S}_{t+1}=\mathcal{S}$ and $\mathbf{G}^{-1}_{\mathcal{S}_{t+1}}=\mathbf{G}^{-1}_{\mathcal{S}}$
\end{algorithmic}
\end{algorithm}

%% file: reconstruction.tex
Extending our derivation in sampling, we propose a biased signal reconstruction method with lower variance and complexity than the conventional LS estimator.

\subsection{Graph Filter Sub-matrix-based Reconstruction}

Recall that the LS estimator \cite{AO} is $\hat{\mathbf{x}}_{\text{LS}} = {\mathbf{V}_K}[ {{{(\mathbf{C}{\mathbf{V}_K})}^{\top}}\mathbf{C}{\mathbf{V}_K}} ]^{-1}{(\mathbf{CV}_K)^{\top}}\mathbf{y}_\mathcal{S}$, where $\mathbf{y}_\mathcal{S}$ is an observed noise-corrupted signal.
The LS reconstruction requires computation of the first $K$ eigenvectors of the graph operator $\mathbf{L}$ and a matrix inverse operation, {and it achieves lowest MSE value among unbiased recovery solution.}
We next derive a biased signal reconstruction method.

\begin{proposition}
Adding a small shift ${\beta}>0$ into the LS solution, we obtain a biased estimator
$\hat{\mathbf{x}} = {\mathbf{V}_K}[ {{{(\mathbf{C}{\mathbf{V}_K})}^{\top}}\mathbf{C}{\mathbf{V}_K}}+{\beta}\mathbf{I} ]^{-1}{(\mathbf{CV}_K)^{\top}}\mathbf{y}_\mathcal{S}$,
which can be approximated by
\begin{equation}
\label{proposition3}
\begin{split}
\hat{\mathbf{x}}=\mathbf{T}^{\emph{\textrm{FGFT}}}_{\mathcal{VS}}{\mathbf{H}}^{-1}_{\mathcal{S}}\mathbf{y}_{\mathcal{S}}
\end{split}
\end{equation}
where $\mathbf{T}^{\emph{\textrm{FGFT}}}$ was obtained during GFS sampling using $\mathbf{T}^{\textrm{FGFT}}=\tilde{\mathbf{V}}_K\tilde{\mathbf{V}}^{\top}_K$, and ${\mathbf{H}}=\mathbf{T}^{\emph{\textrm{FGFT}}}+{\beta}\mathbf{I}$.
\end{proposition}
\begin{proof}
Similar derivation in \eqref{shift Neumann series} and \eqref{trace property} can show that the shifted LS solution is equal to
\begin{align}\label{recovery1}
&\hspace{0cm}\hat{\mathbf{x}}={\mathbf{V}_K}\left[\sum\limits_{l = 0}^{\infty} {\sum\limits_{d = 0}^{l} \tbinom{l}{d}(1-{\beta})^{l-d} (\mathbf{-\Psi})^{d} }\right]{(\mathbf{CV}_K)^{\top}}\mathbf{y}_\mathcal{S}\nonumber\\
&\hspace{0.3cm}\mathop  = \limits^{(a)}\mathbf{T}\left[ \sum\limits_{l = 0}^{\infty}{\sum\limits_{d = 0}^{l} \tbinom{l}{d}(1-{\beta})^{l-d} (\mathbf{-P})^{d} }\right]\mathbf{C}^{\top}\mathbf{y}_\mathcal{S}\\
&\hspace{0.3cm}=\mathbf{T}\sum\limits_{l = 0}^{\infty}\left[\mathbf{I}-(\mathbf{P}+{\beta}\mathbf{I})\right]^{l}\mathbf{C}^{\top}\mathbf{y}_\mathcal{S}\nonumber
 \end{align}
where $\mathbf{P}={ {\mathbf{C}}^{\top}}\mathbf{C}{\mathbf{V}_K}{\mathbf{V}^{\top}_K}$ and (a) holds since ${\mathbf{V}_K}\mathbf{\Psi}^d{\mathbf{V}^{\top}_K}={\mathbf{V}_K}({\mathbf{V}^{\top}_K}{ {\mathbf{C}}^{\top}}\mathbf{C}{\mathbf{V}_K})...
({\mathbf{V}^{\top}_K}{{\mathbf{C}}^{\top}}\mathbf{C}{\mathbf{V}_K}){\mathbf{V}^{\top}_K}
={\mathbf{V}_K}{\mathbf{V}^{\top}_K}\mathbf{P}^{d}$.

From the definition of $\mathbf{C}$, we know that ${\mathbf{C}^{\top}}\mathbf{C} = {\left[ \begin{array}{l}
 {\mathbf{I}_\mathcal{S}}~~~ {\bf{0}} \\
 {\bf{0}}~~~~{\kern 1pt} {\bf{0}} \\
 \end{array} \right]}$ under appropriate permutation, which leads to $\mathbf{P} = \left[ \begin{array}{l}
{\mathbf{T}_\mathcal{SV}} \\
~\mathbf{0} \\
 \end{array} \right]$.
 Hence, $
\mathbf{I} - (\mathbf{P}+{\beta}\mathbf{I}) = \left[ \begin{array}{l}
{\mathbf{I}}- ({\mathbf{T}_\mathcal{S}}+{\beta}\mathbf{I})~~-{\mathbf{T}_{\mathcal{S}{\mathcal{S}^c}}} \\
~~~~~~\mathbf{0}~~~~~~~~~~~(1-{\beta}){\mathbf{I}} \\
 \end{array} \right]
 $,
 resulting in
\begin{equation}\label{structure}
\begin{split}
\left[\mathbf{I} - (\mathbf{P}+{\beta}\mathbf{I})\right]^{l} = \left[ \begin{array}{l}
\left[\mathbf{I} - (\mathbf{T}_{\mathcal{S}}+{\beta}\mathbf{I})\right]^{l}~~ \bullet\\
~~~~~~~~ \mathbf{0}~~~~~~~~(1-{\beta})^{l}{\mathbf{I}}
\end{array} \right]
\end{split}
\end{equation}
where ``$\bullet$'' denotes a nonzero matrix whose dimension is $M\times (N-M)$.

Then, \eqref{recovery1} can be written as
 \begin{equation}\label{recovery2}
 \begin{split}
&\hspace{0cm}\hat{\mathbf{x}}=\mathbf{T}\sum\limits_{l = 0}^{\infty} \left[ \begin{array}{l}
\left[\mathbf{I} - (\mathbf{T}_{\mathcal{S}}+{\beta}\mathbf{I})\right]^{l}~~ \bullet\\
~~~~~~~~ \mathbf{0}~~~~~~~~(1-{\beta})^{l}{\mathbf{I}}
\end{array} \right]\left[ \begin{array}{l}
\mathbf{y}_\mathcal{S}   \\
{\kern 1pt}{\kern 1pt}\mathbf{0} \\
 \end{array} \right]\\
 &\hspace{0.3cm}=\mathbf{T}_{\mathcal{VS}}\sum\limits_{l = 0}^{\infty}\left[\mathbf{I} - (\mathbf{T}_{\mathcal{S}}+{\beta}\mathbf{I})\right]^{l}\mathbf{y}_\mathcal{S}\\
 &\hspace{0.3cm}\mathop=\limits^{(a)}\mathbf{T}_{\mathcal{VS}}(\mathbf{T}_{\mathcal{S}}
 +{\beta}\mathbf{I})^{-1}\mathbf{y}_\mathcal{S}
 \end{split}
 \end{equation}
where equality $(a)$ holds from a similar derivation of \eqref{back transformation} with a different shift.
Finally, we approximate $\mathbf{T}$ by the $\mathbf{T}^{\textrm{FGFT}}$.
\end{proof}
\begin{figure*}[htbp]
    \centering
     \subfigure[Reconstruction MSE in {\textbf{G1}} at 10dB]{
    \begin{minipage}{5.5cm}
    \centering
        \includegraphics[width=2.2in,height=2.2in]{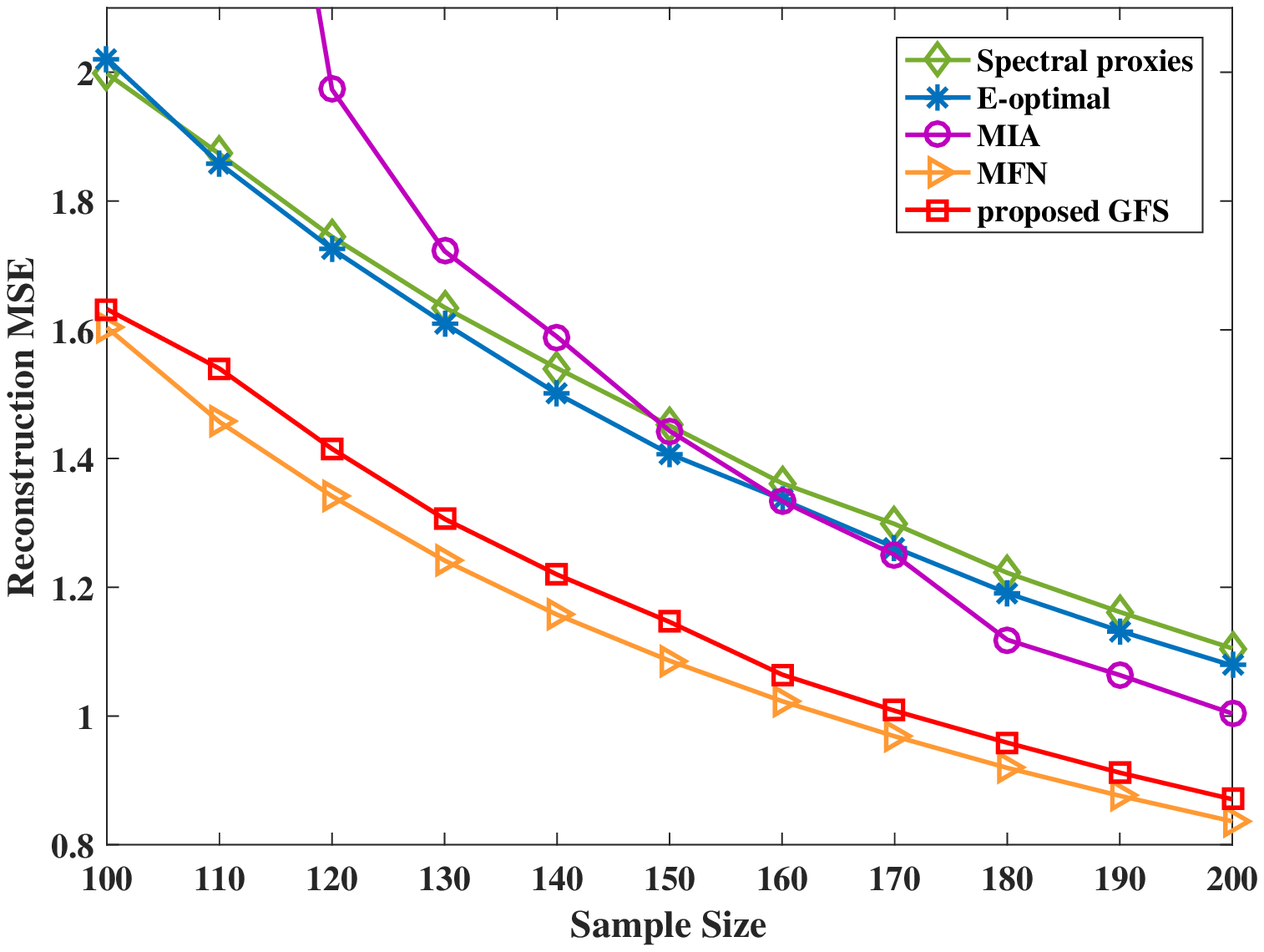}
    \end{minipage}%
    }
     \subfigure[Reconstruction MSE in {\textbf{G1}} at 0dB]{
    \begin{minipage}{5.5cm}
    \centering
        \includegraphics[width=2.2in,height=2.2in]{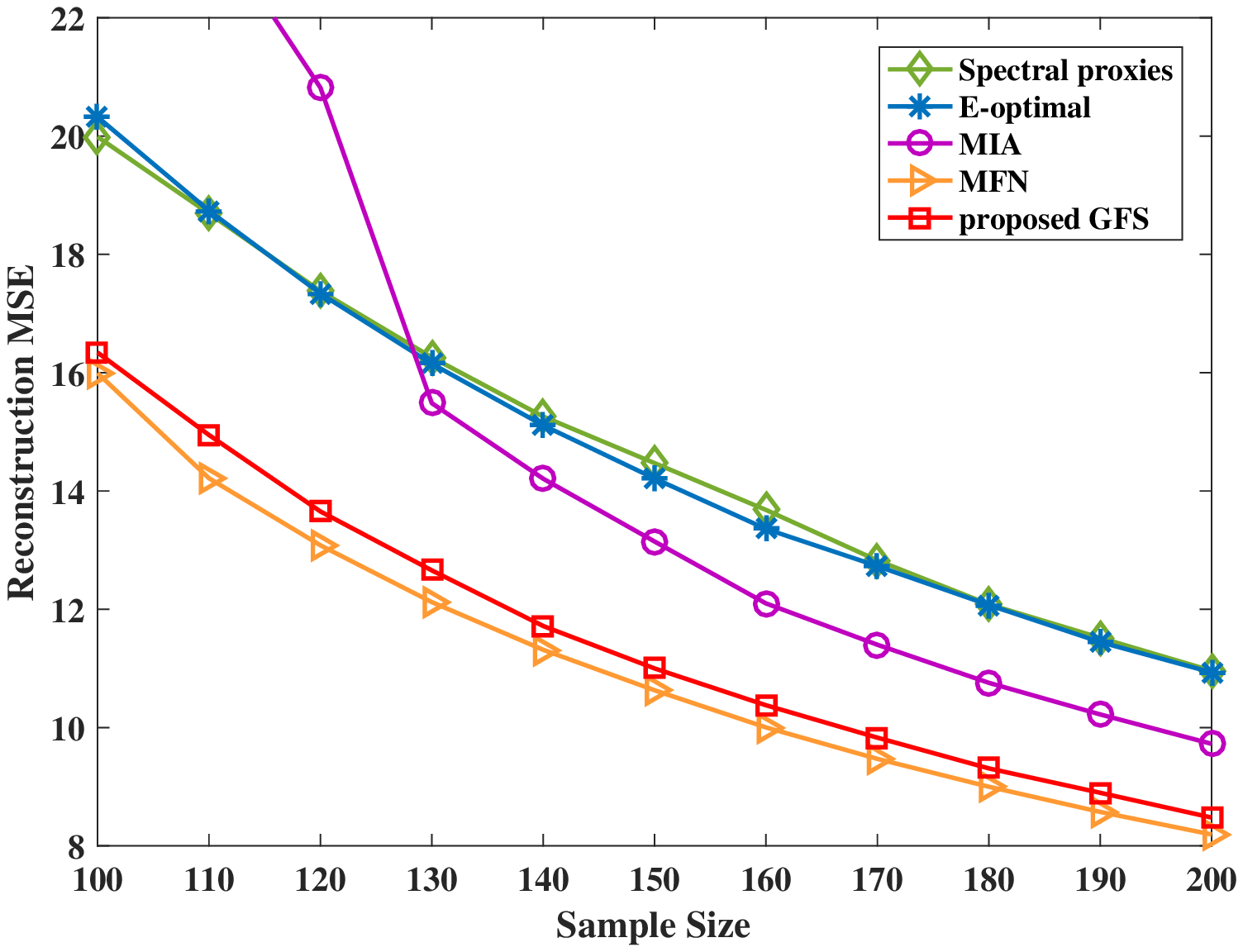}
    \end{minipage}%
    }
    \subfigure[Reconstruction MSE in {\textbf{G2}} at 10dB]{
    \begin{minipage}{5.5cm}
    \centering
        \includegraphics[width=2.2in,height=2.2in]{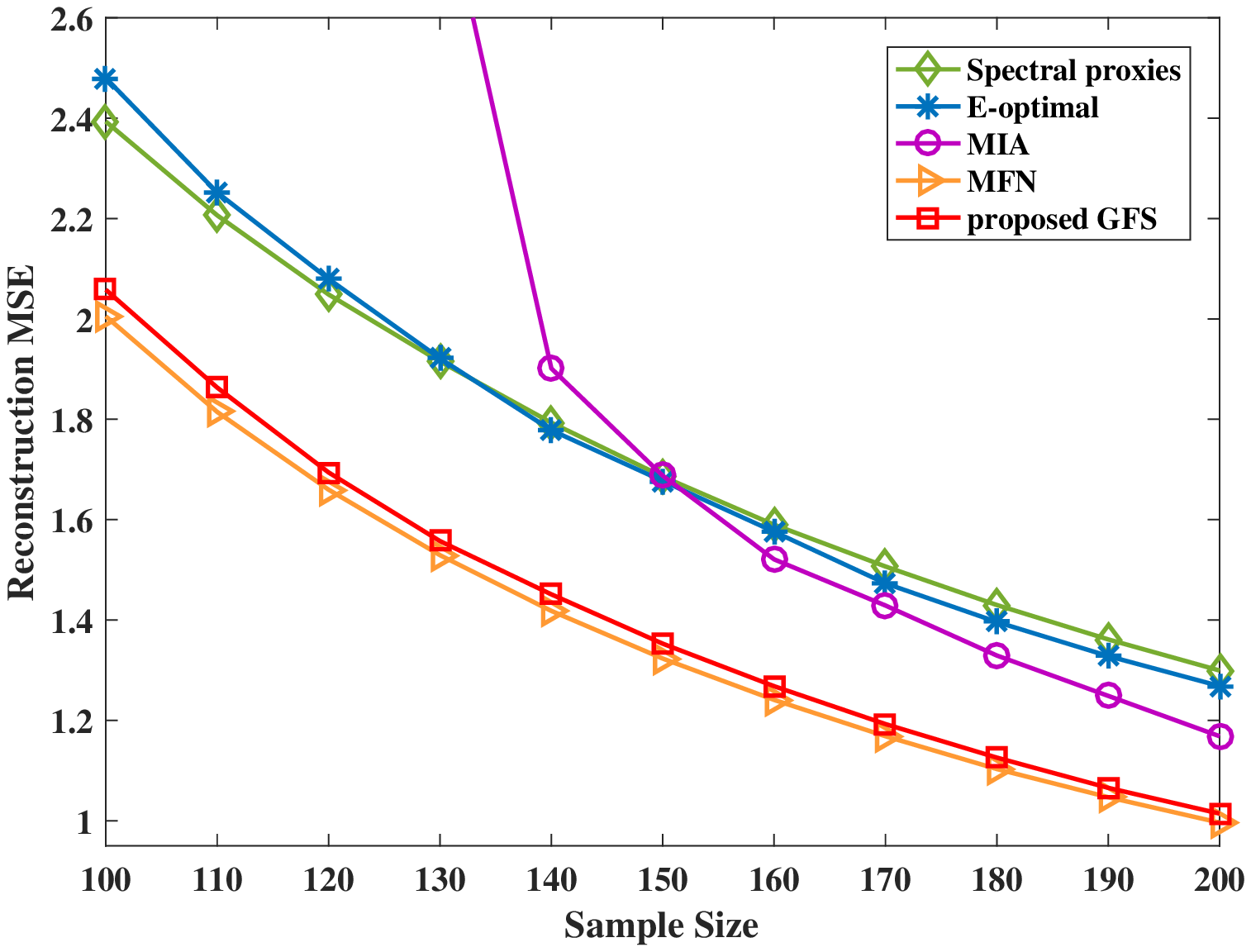}
    \end{minipage}%
    }
     \centering
     \subfigure[Reconstruction MSE in {\textbf{G2}} at 0dB]{
    \begin{minipage}{5.5cm}
    \centering
        \includegraphics[width=2.2in,height=2.2in]{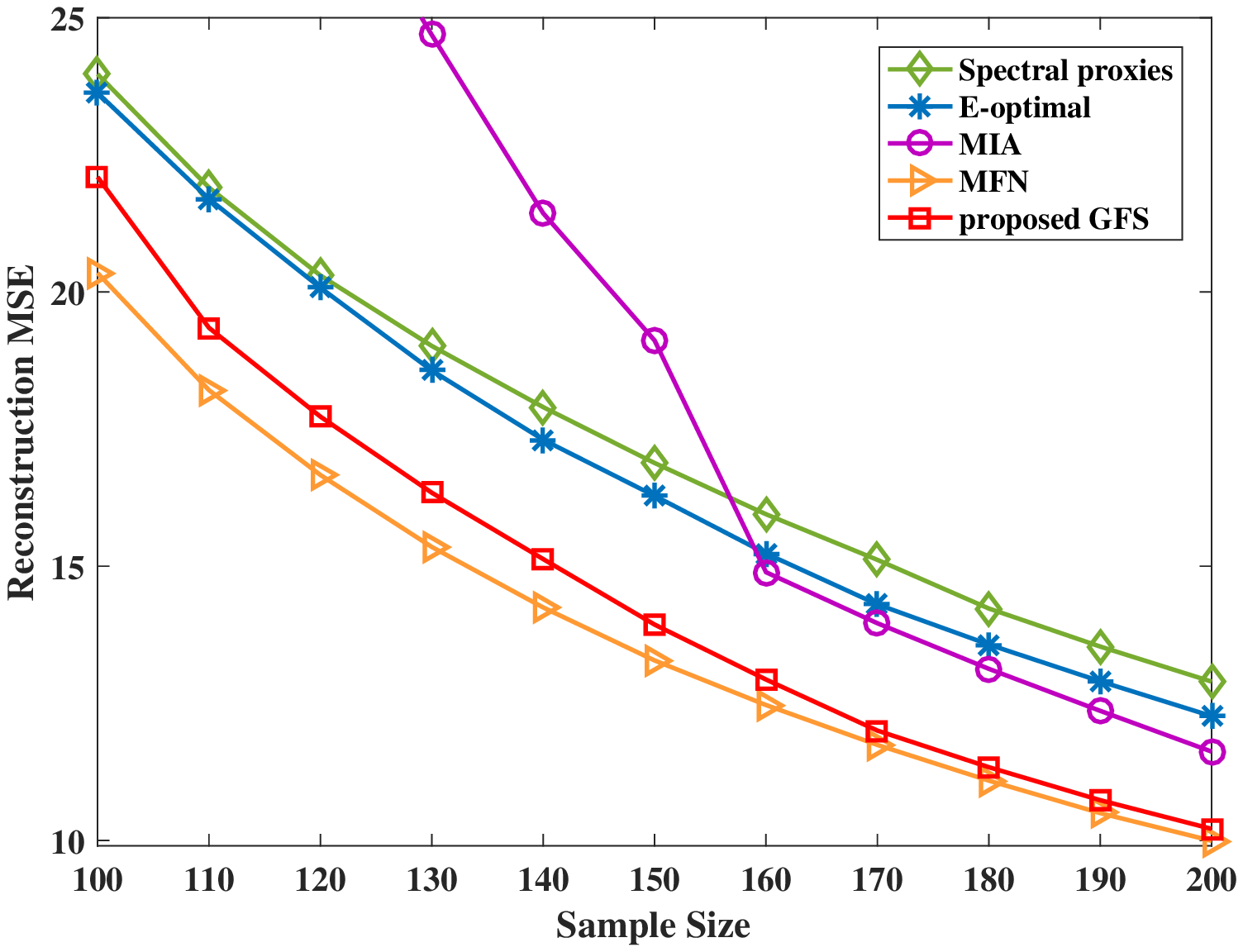}
    \end{minipage}%
    }
     \subfigure[Reconstruction MSE in {\textbf{G3}} at 10dB]{
    \begin{minipage}{5.5cm}
    \centering
        \includegraphics[width=2.2in,height=2.2in]{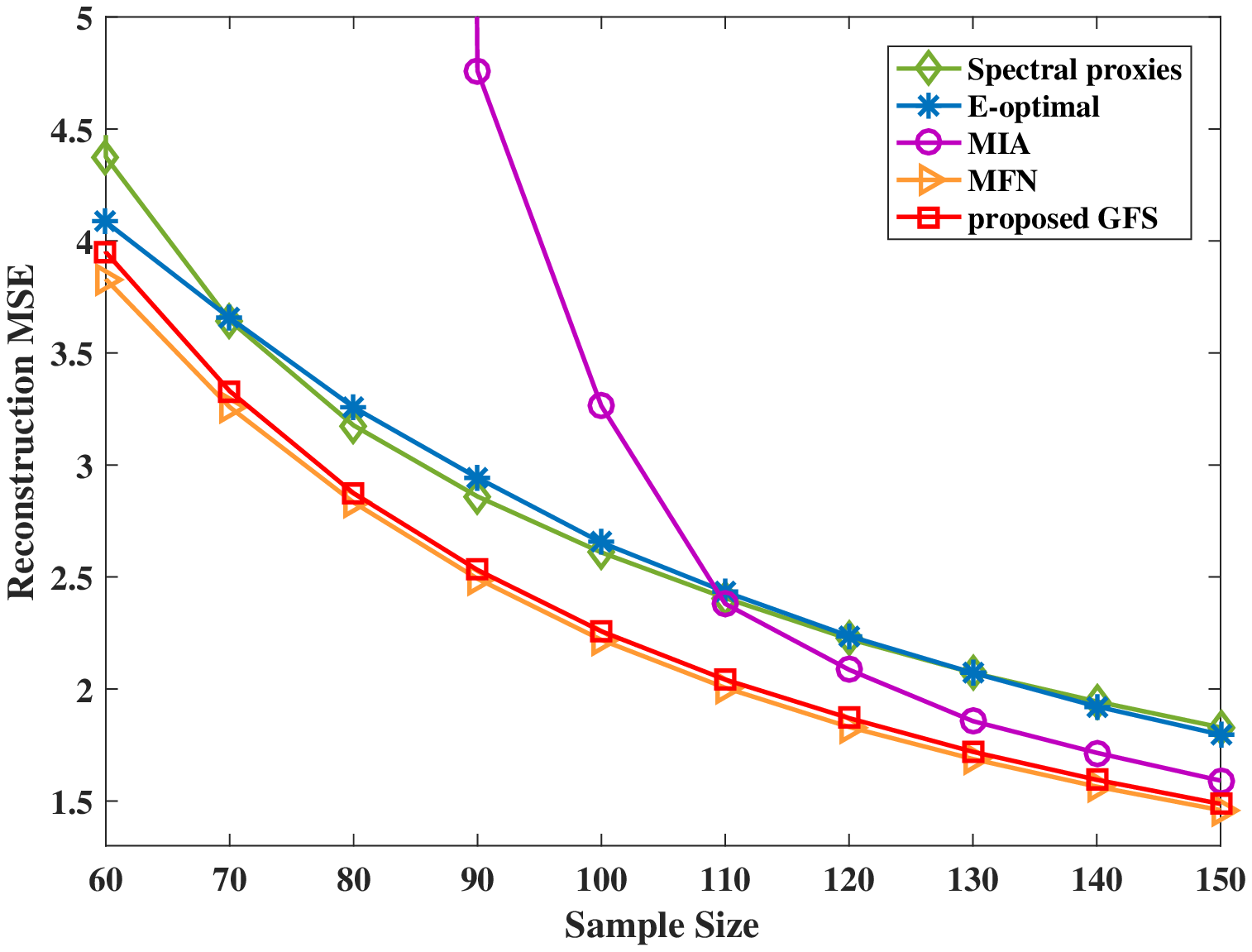}
    \end{minipage}%
    }
    \subfigure[Reconstruction MSE in {\textbf{G3}} at 0dB]{
    \begin{minipage}{5.5cm}
    \centering
        \includegraphics[width=2.2in,height=2.2in]{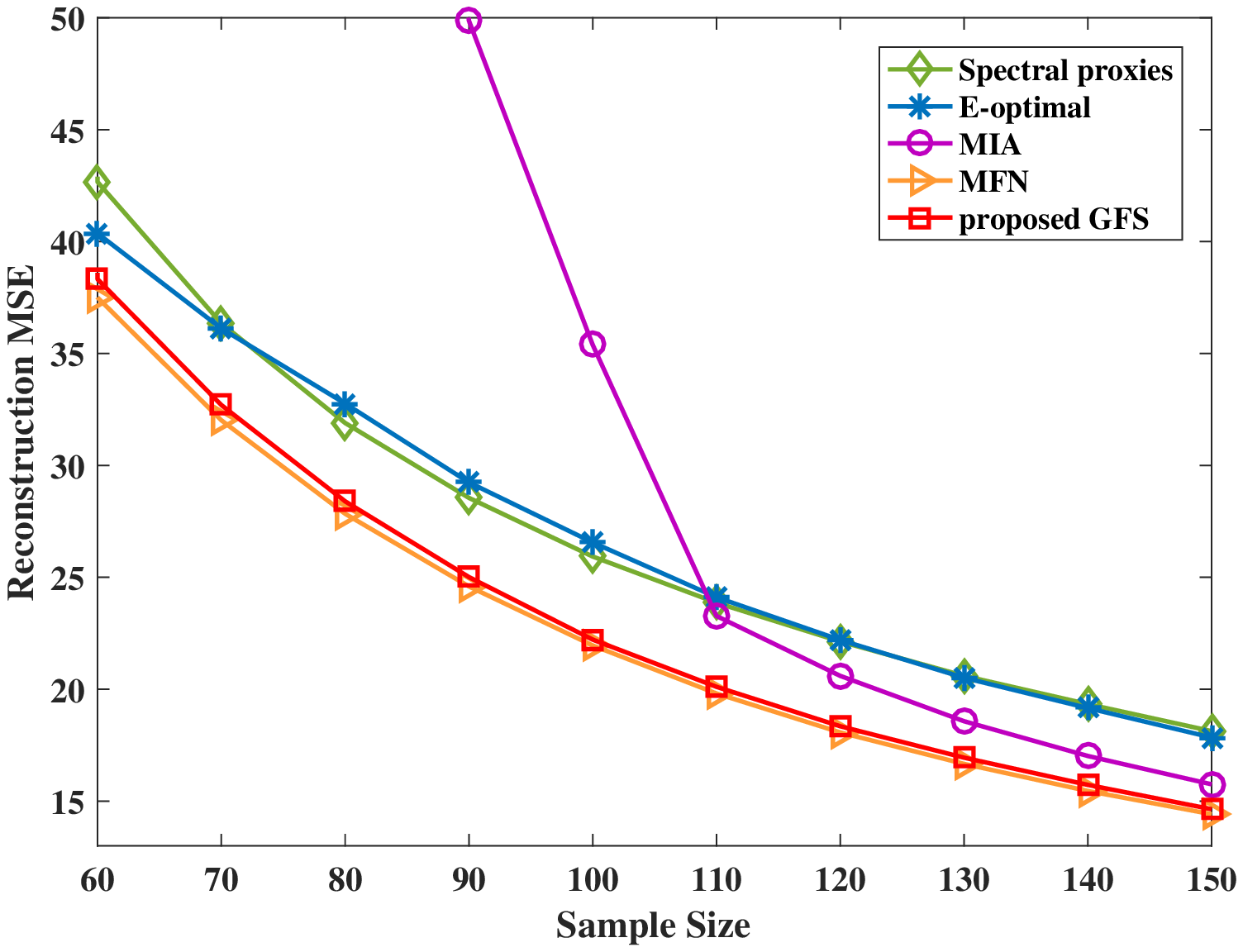}
    \end{minipage}%
    }
\caption{Reconstruction MSE of different sampling algorithms at different SNRs and graph types, where the original graph signals are all recovered from samples by the LS reconstruction.}
\label{staticMSE}
\end{figure*}
{\textbf{Remark}}: {If we choose $\beta=\mu$, then $\mathbf{H}=\mathbf{G}$. All involved matrices in \eqref{proposition3} have been obtained in sampling.
However, the value for $\mu$ has been designed in the sampling procedure based on the condition number constraint.
If $\beta\neq\mu$, we can customize a $\beta$ based on the bias-variance tradeoff (to be discussed), at the cost of computing the inverse $\mathbf{H}^{-1}_{\mathcal{S}}$ once.

We next prove the robustness of the proposed estimator $\mathbf{\hat{x}}$ and propose one strategy to compute the value of $\beta$.}

\vspace{-0.08in}
\subsection{Robustness Analysis}

Assume that the signal has the same energy for different SNRs, and noise $\mathbf{n}$ is i.i.d. with zero mean and variance $\omega^2$ which varies with SNR.
We write the following proposition.

\begin{proposition}
\label{robustness}
Given $\mathbf{T}^{\emph{\textrm{FGFT}}}$ sufficiently approximates the ideal LP filter, MSE of the proposed GFS reconstruction is:
\begin{equation}\label{GFSMSE}
\begin{split}
&\hspace{-0.2cm} {\bf{E}}\left\| {{\bf{\hat x}} - {\bf{x}}} \right\|_2^2
=\sum^K_{i=1}\left(1+\frac{\sigma_i}{\beta}\right)^{-2}
(\mathbf{u}^{\top}_i\tilde{\mathbf{x}}_K)^2+
\omega^2\sum^{K}_{i=1}\frac{\sigma_i}{(\sigma_i+\beta)^2}
\end{split}
\end{equation}
and the MSE of the LS solution is:
\begin{align}\label{LSMSE}
&\hspace{0cm} {\mathbb{E}}\left\| {{\bf{\hat x}_{\emph{\textrm{LS}}}} - {\bf{x}}} \right\|_2^2
= {\omega ^2}\sum\limits_{i = 1}^K {\frac{1}{{ {\sigma _i}}}}.
\end{align}
where $(\sigma_i,\mathbf{u}_i)$ is the $i$-th eigen-pair of the matrix $\mathbf{\Psi}  = {\left( {\mathbf{C}{\mathbf{V}_K}} \right)^{\top}}\mathbf{C}{\mathbf{V}_K}$.

\end{proposition}
\begin{proof}
The proof is detailed in Appendix \ref{robustnessProof}.
\end{proof}

When the noise variance $\omega^2$ is large enough to dominate MSE, the GFS recovery has lower MSE than the LS solution since $\frac{\sigma_i}{(\sigma_i+\beta)^2}<\frac{1} {\sigma _i}$, $\forall i$.
Therefore the proposed method is more robust to large noise than the LS reconstruction.

{The RHS of formula \eqref{GFSMSE} is consisted of bias and variance of estimator $\mathbf{\hat{x}}$ respectively.
For the bias part, $\mathbf{u}_i$ and $\mathbf{\tilde{x}}_K$ are constant in reconstruction, thus smaller $\beta$ will bring smaller bias.
For the variance part, smaller $\beta$ will bring larger variance.
Hence, given the noise variance $\omega^2$, the optimal $\beta$ can be designed to balance the bias-variance tradeoff to achieve the lowest MSE.
Instead of computing the gradient of function \eqref{GFSMSE}, we propose an empirical and practical strategy to design $\beta$ based on some prior information.
For lower MSE, $\beta$ should be comparable to $\sigma_i$ to combat large noise and to maintain reasonable bias.}
We observe that the average value of $\sigma_i$ is
\begin{equation}
\begin{split}
\bar{\sigma}=\frac{1}{K}\sum^{K}_{i=1}\sigma_i=\frac{1}{K}\text{tr}\left[{\left( {\mathbf{C}{\mathbf{V}_K}} \right)^{\top}}\mathbf{C}{\mathbf{V}_K}\right]=\frac{1}{K}\text{tr}\left(\mathbf{T}_{\mathcal{S}}\right)
\end{split}
\end{equation}

We adopt $\mathbf{T}^{\text{FGFT}}$ to approximate $\mathbf{T}$ and select $\beta$ to be the lower bound of the above average value
\begin{equation}\label{beta}
\begin{split}
&\hspace{0cm}\beta=\frac{1}{K}\sum_{i\in \mathcal{J}} \mathbf{T}^{\text{FGFT}}_{ii}
\end{split}
\end{equation}
where $|\mathcal{J}|=M$ and $\mathcal{J} \subset \mathcal{V}$. $\forall i\in\mathcal{J}$ and $\forall j\in \mathcal{J}^c$, $\mathbf{T}^{\text{FGFT}}_{ii}\leq\mathbf{T}^{\text{FGFT}}_{jj}$.
We will validate the performance of this selection of $\beta$ via simulations.

%% file: experiments.tex
\begin{figure}
\begin{center}
\includegraphics[width=170pt,height=160pt]{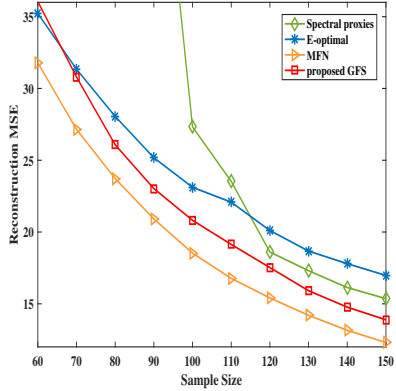}
\caption{Reconstruction MSE of different sampling algorithms at $\textrm{SNR}=0\textrm{dB}$ on the Minnesota graph with $K=50$}
\label{MinnesotaMSE}
\end{center}
\end{figure}
\begin{figure*}[htbp]
    \centering
     \subfigure[Original signal, $K=50$]{
    \begin{minipage}{5.5cm}
    \centering
        \includegraphics[width=2.0in,height=1.8in]{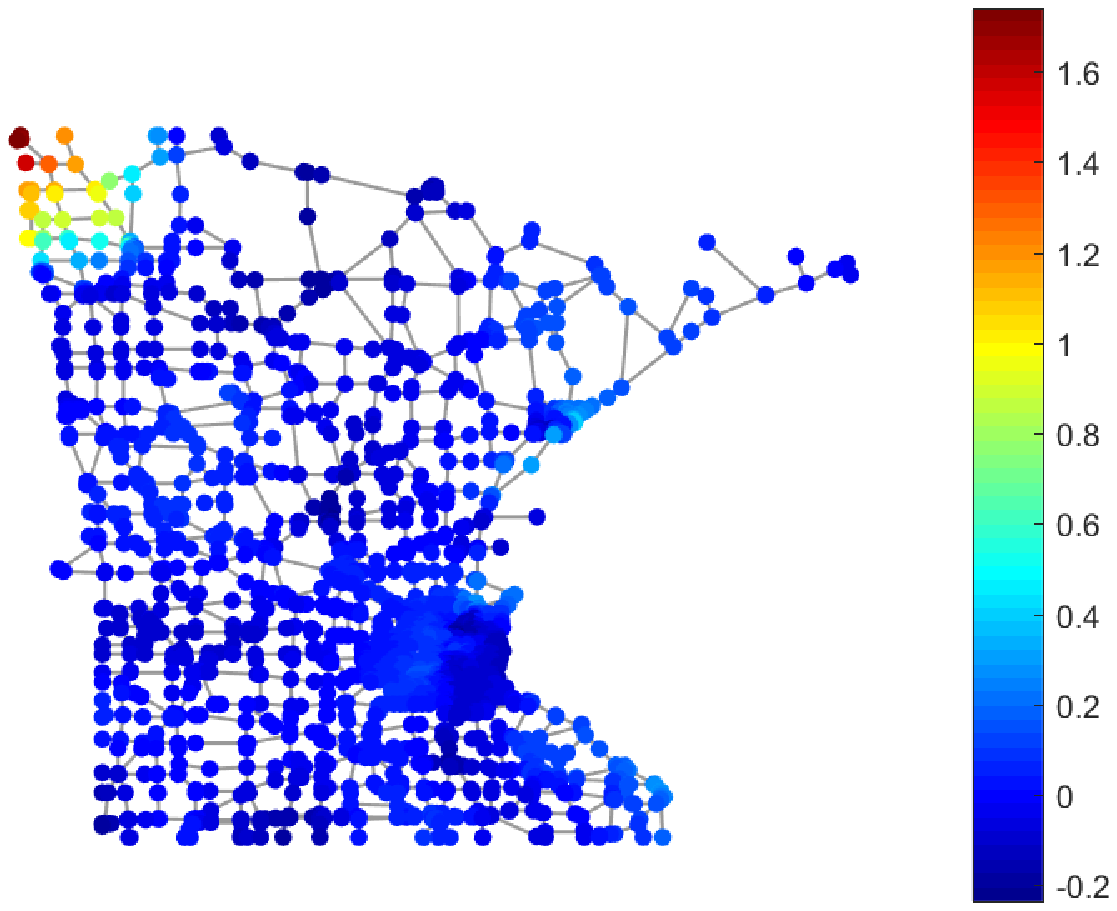}
    \end{minipage}%
    }
     \subfigure[Spectral proxies \cite{AO}, MSE=67.4]{
    \begin{minipage}{5.5cm}
    \centering
        \includegraphics[width=2.0in,height=1.8in]{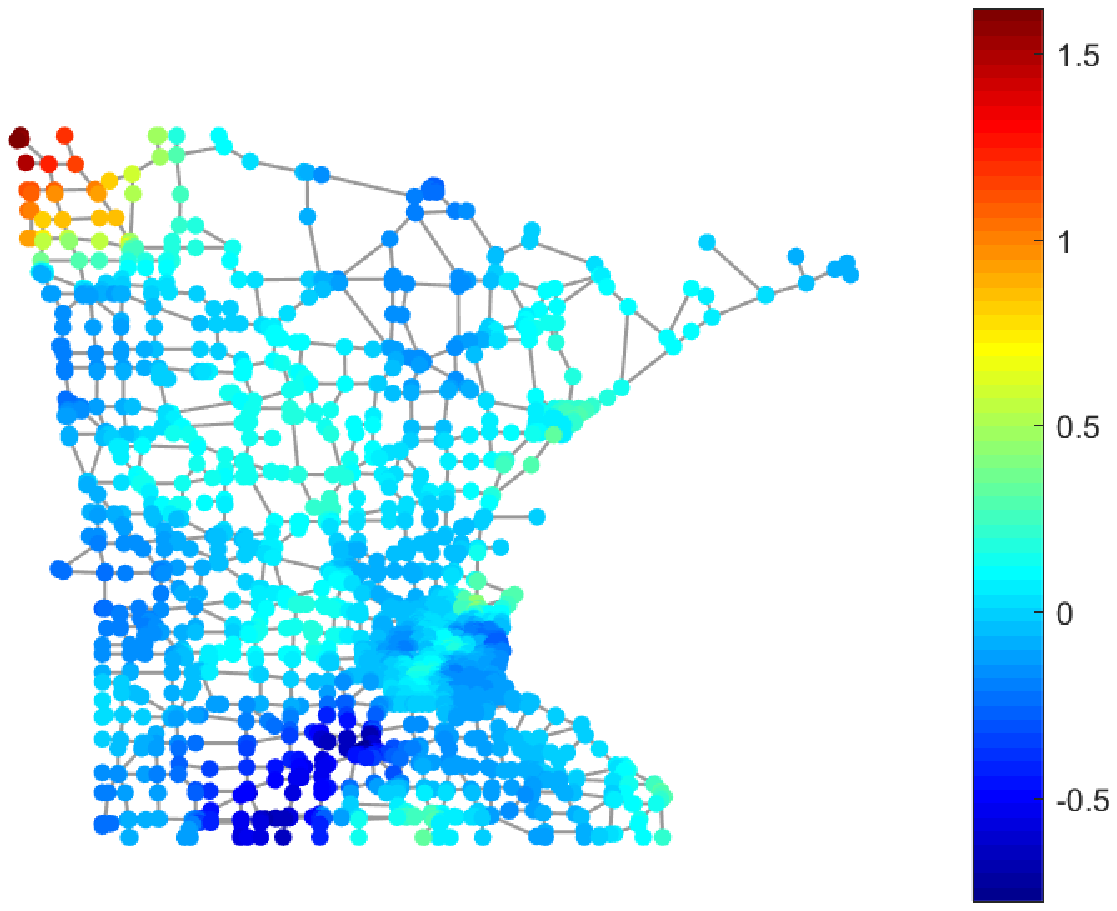}
    \end{minipage}%
    }
    \subfigure[E-optimal \cite{SCsampling}, MSE=22.3]{
    \begin{minipage}{5.5cm}
    \centering
        \includegraphics[width=2.0in,height=1.8in]{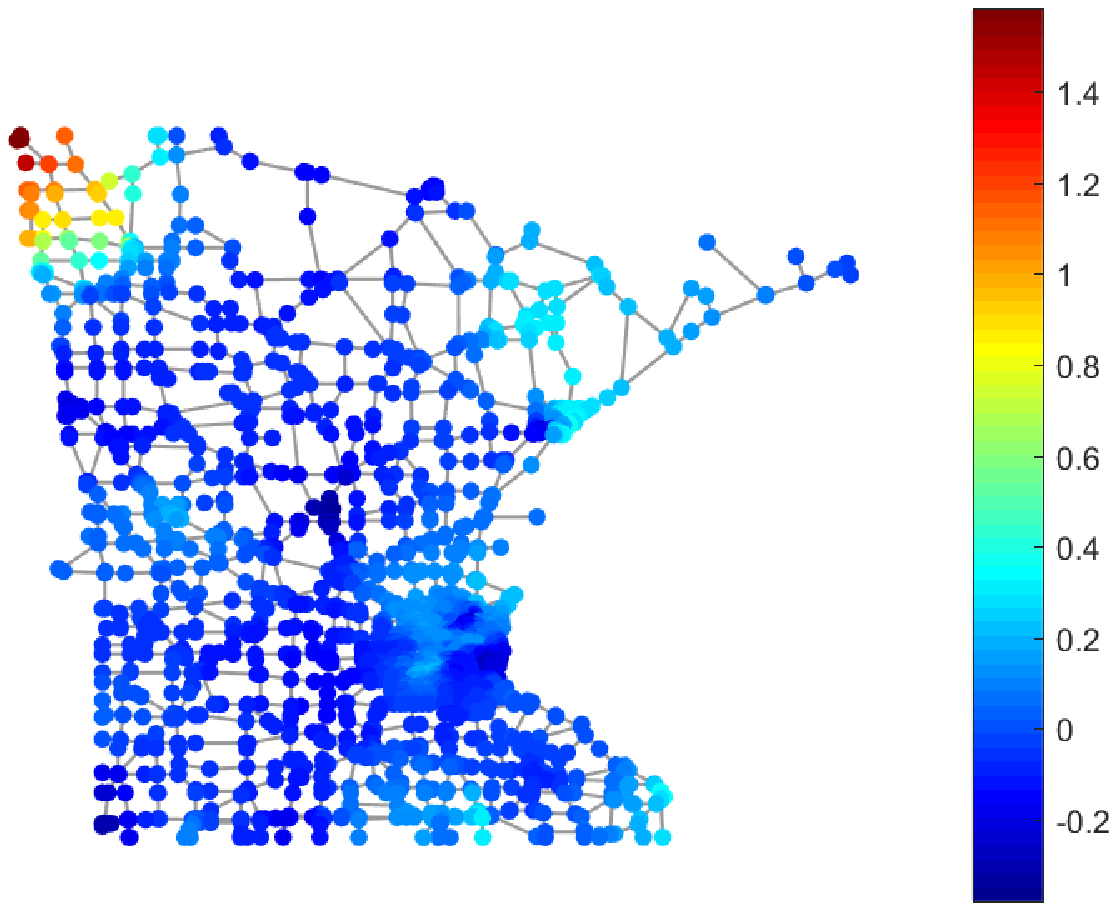}
    \end{minipage}%
    }
     \centering
     \subfigure[MIA \cite{SPL}, MSE=1525]{
    \begin{minipage}{5.5cm}
    \centering
        \includegraphics[width=2.0in,height=1.8in]{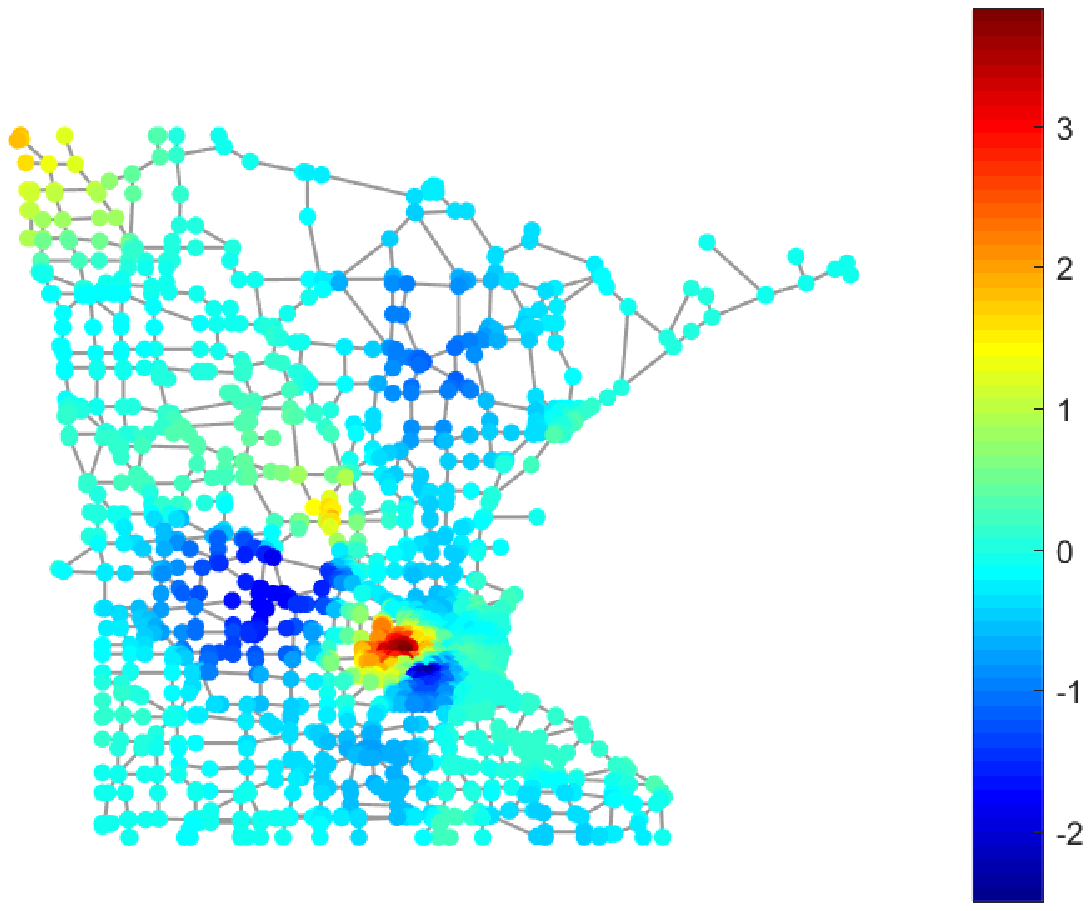}
    \end{minipage}%
    }
     \subfigure[MFN \cite{Uncertainty}, MSE=26.2]{
    \begin{minipage}{5.5cm}
    \centering
        \includegraphics[width=2.0in,height=1.8in]{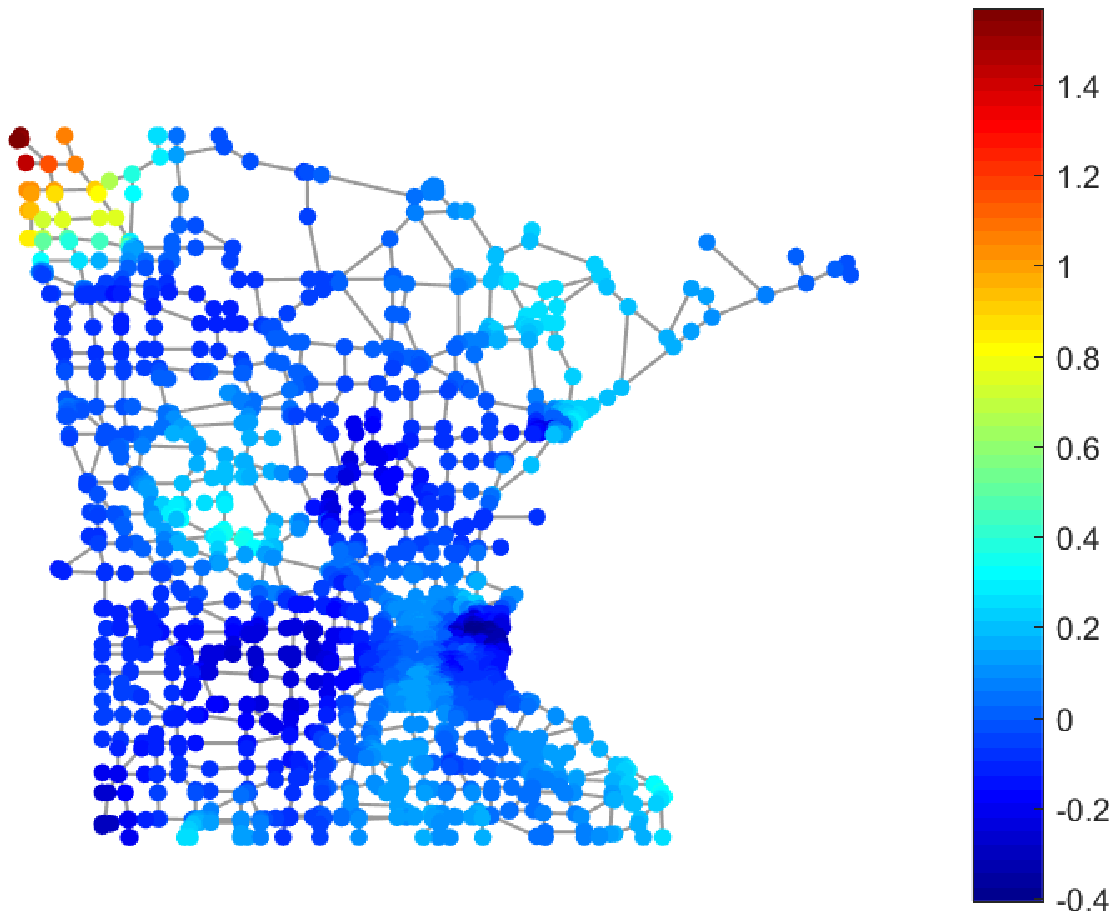}
    \end{minipage}%
    }
    \subfigure[Proposed GFS, MSE=16.5]{
    \begin{minipage}{5.5cm}
    \centering
        \includegraphics[width=2.0in,height=1.8in]{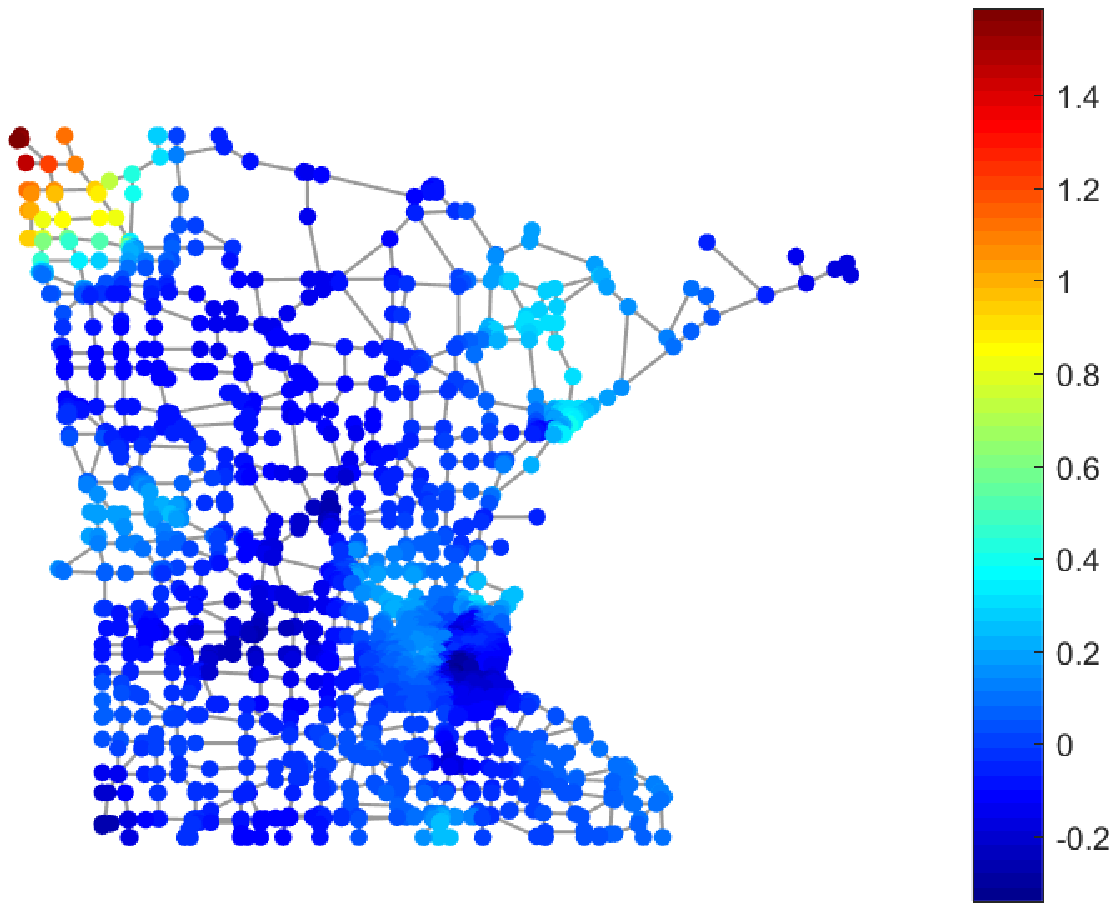}
    \end{minipage}%
    }
\caption{One specific demonstration of the reconstruction MSE of different sampling algorithms at $\textrm{SNR}=0\textrm{dB}$ on Minnesota graph with $M=90$}
\label{MinnesotaDemo}
\end{figure*}
Our proposed sampling and reconstruction strategies are evaluated via experimental simulations.
All experiments were performed in MATLAB R2017b, running on a PC with Intel Core I3 3.7 GHz CPU and 16GB RAM.
The simulated artificial graphs \footnote{All of these graphs are generated using GSP open source in \cite{GSPtool}.} are described as follows:

{\textbf{(G1)}} Community graphs with 1000 nodes and 31 communities;

{\textbf{(G2)}} Sensor graphs with 1000 nodes.

{\textbf{(G3)}} Hype-cube graphs with 1002 nodes in 3-dimension.

We also perform experiments on the Minnesota network, which is a real-world graph with $2642$ nodes.
Artificial graph signals are assumed to be bandlimited with bandwidth $K=50$.
They are constructed by generating appropriate GFT coefficients: the non-zero GFT coefficients are randomly generated from the distribution $\mathcal{N}\left(1,0.5^{2}\right)$, and the coefficients after $K=50$ are zeros.
The generated graph signals are corrupted by additional white Guassian noise (AWGN) with different signal-to-noise ratios (SNRs).

\hspace{-0.1in}
\subsection{Static Global Sampling}
\subsubsection{Reconstruction MSE evaluation}
The metric used to evaluate a sampling scheme is MSE of the reconstructed signal with respect to the ground truth signal.
We compare our proposed GFS sampling algorithm to competing schemes that employ other criteria: SP in \cite{AO}, E-optimal in \cite{SCsampling}, MIA in \cite{SPL} and MFN in \cite{Uncertainty}. 
The truncation degree in MIA is $L=10$.
The SGWT toolbox \cite{GSPtool} is adopted to approximate the ideal LP filter in the MIA method, where $p=25$ and $\alpha=30$ \cite{semisupervised}.
For SP, the approximation order is $k=10$.
As described in Section \ref{sec:AOptimal}, for GFS sampling, the shift parameter $\mu$ is set to be $1/(\kappa_0-1)$, where we set $\kappa_0=100$ as the condition number constraint.
The number of Given rotations matrices is $J=6N\text{log}N$ for {\bf{G1}}--{\bf{G3}} and $J=30N\text{log}N$ for the Minnesota graph.

{Fig.\;\ref{staticMSE} shows the reconstruction MSE of different sampling strategies in terms of sample size in {\bf{G1}}--{\bf{G3}} at different SNRs. }
As shown in Fig.\;\ref{staticMSE}, our proposed GFS achieves lower MSE than three other competing sampling strategies, and closely approximates MFN's performance in community, sensor and cube graphs at different SNRs.
{As previously discussed, MFN greedily minimizes MSE one node at a time directly, thus it needs to compute the first $K$ eigenvectors once and perform matrix inverse for each metric evaluation, as detailed in Section \ref{sec:greedy}.}
In contrast, our proposed strategy can obtain comparable performance with significantly lower complexity.

Fig.\;\ref{MinnesotaMSE} shows the reconstruction MSE of different sampling methods for the real-world Minnesota network, showing that our proposed sampling also outperforms SP and E-optimal sampling, where the performance of MIA is too poor to be shown.
For visualization, we further present one specific demonstration in Fig.\,\ref{MinnesotaDemo}.
Fig.\,\ref{MinnesotaDemo} (a) and (b)--(f) show the original graph signal and the interpolated  signal of different sampling methods, respectively.
Fig.\,\ref{MinnesotaDemo} (b)--(f) show that the proposed GFS achieves the lowest MSE value among competing strategies, and the reconstructed signal is visually smoother with respect to the graph topology compared with signal recovered by SP and MIA.
\begin{figure}
\begin{center}
\includegraphics[width=220pt,height=150pt]{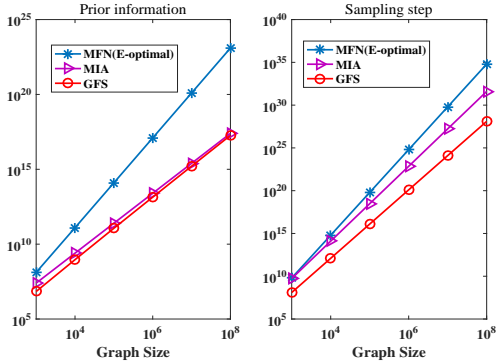}
\caption{Numerical comparison of complexity listed in Table \ref{complexity} with $M=0.05N$.}
\label{complexityFig}
\end{center}
\end{figure}
\begin{figure*}[htbp]
    \centering
     \subfigure[Reconstruction MSE in {\textbf{G1}}]{
    \begin{minipage}{5.5cm}
    \centering
        \includegraphics[width=2.0in,height=2.0in]{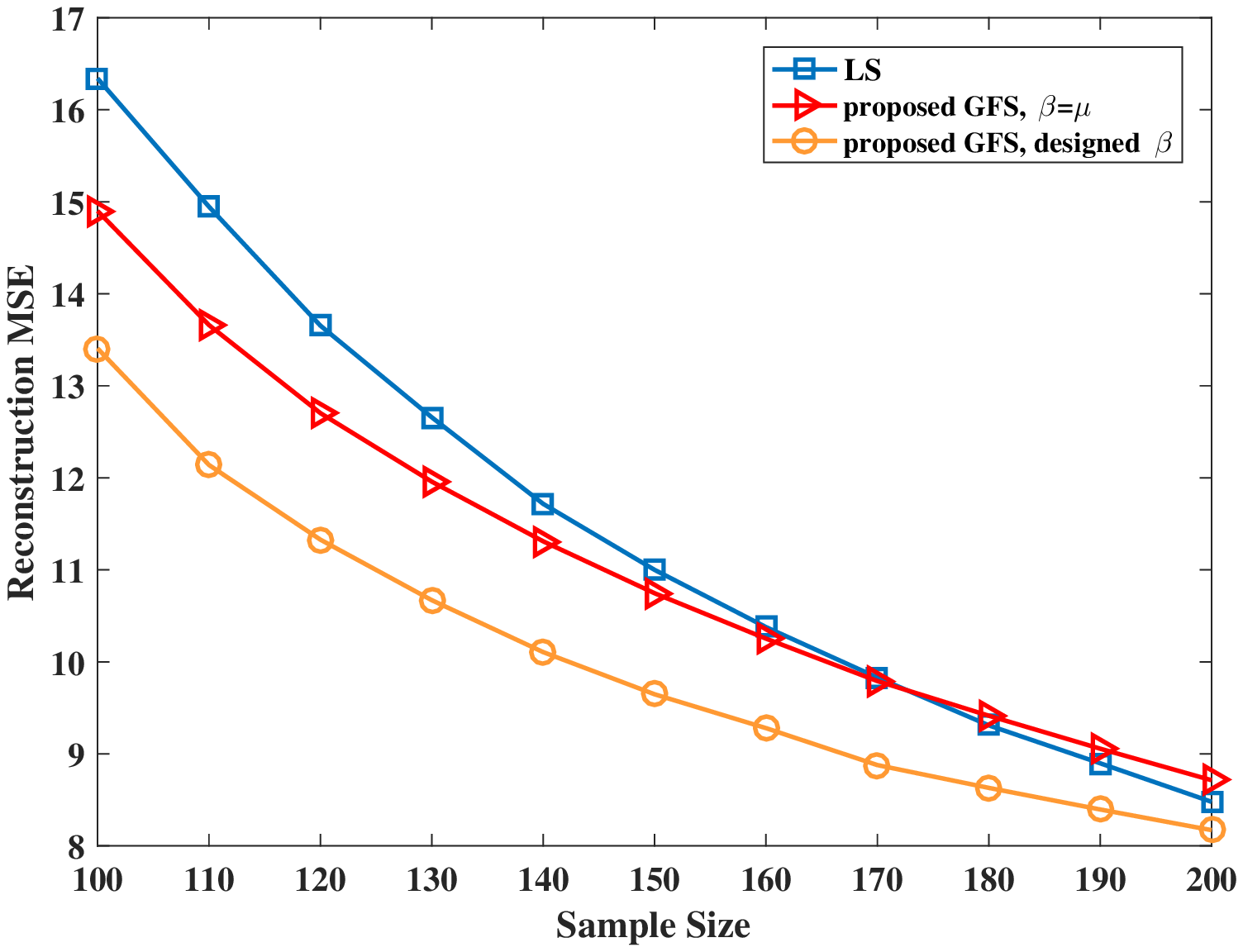}
    \end{minipage}%
    }
     \subfigure[Reconstruction MSE in {\textbf{G2}}]{
    \begin{minipage}{5.5cm}
    \centering
        \includegraphics[width=2.0in,height=2.0in]{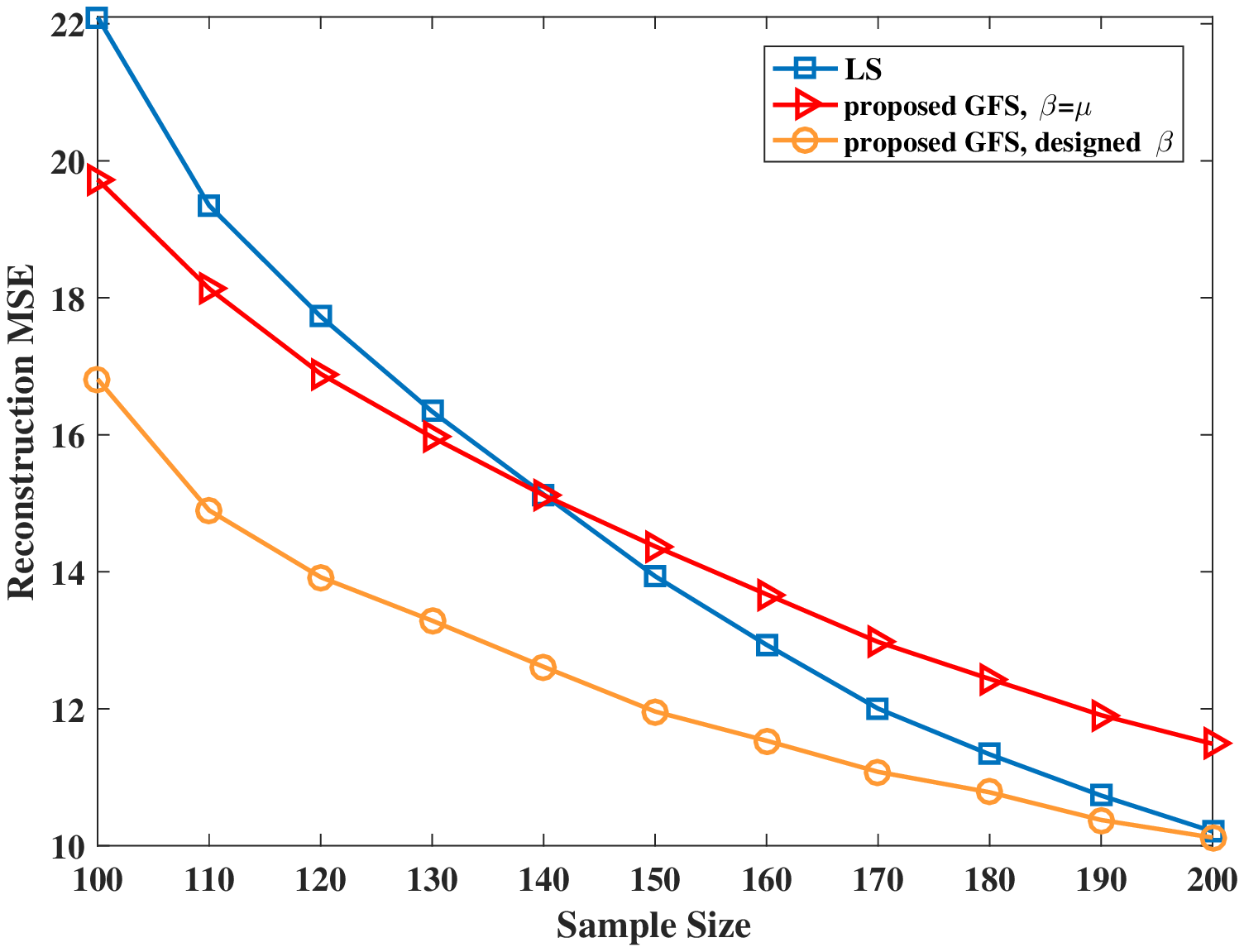}
    \end{minipage}%
    }
    \subfigure[Reconstruction MSE in {\textbf{G3}}]{
    \begin{minipage}{5.5cm}
    \centering
        \includegraphics[width=2.0in,height=2.0in]{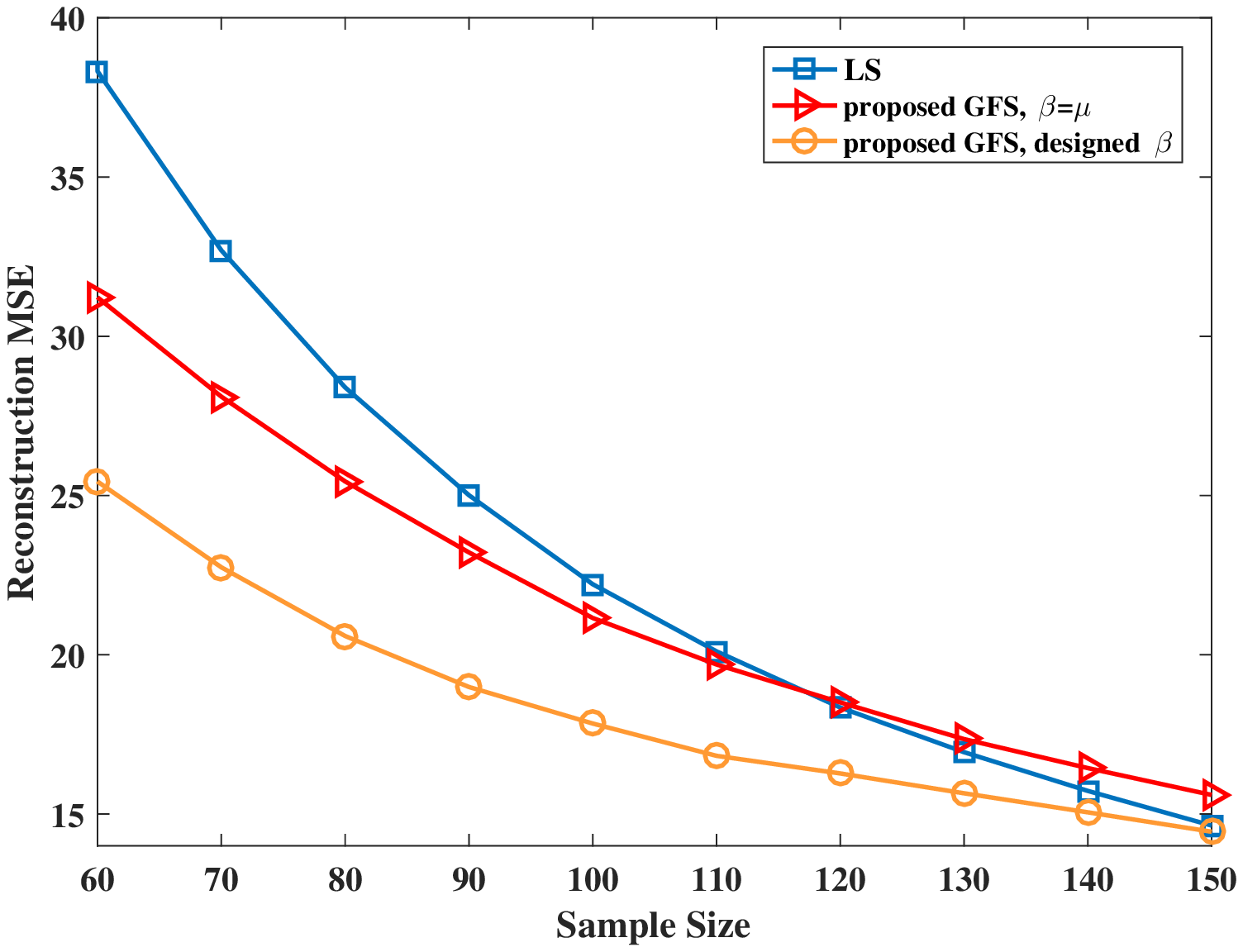}
    \end{minipage}%
    }
\caption{Experimental results of different recovery strategies in different graph types at $\textrm{SNR}=0\textrm{dB}$. The original graph signals are all sampled via GFS algorithm.}
\label{recMSE}
\end{figure*}

\subsubsection{Numerical complexity comparison}
{We evaluate the complexity of GFS via numerical simulations.
Specifically, we compute the numerical complexity order value in Table \ref{complexity} assuming $M=0.05 N$. }
In these experiments, we set $R=10$, $T_1=T_2(k)=100$ and $|\mathcal{E}|=\mathcal{O}(N)$.
Fig.\;\ref{complexityFig} shows the numerical complexity value and shows that GFS has lower complexity than MFN, E-optimal and MIA, especially for large graphs.
\subsubsection{Appropriateness of the shift parameter $\mu$}
The augmented A-optimality function will well approximate the original criterion \eqref{original formulation} if $\mu$ is sufficienty  small.
However, we customize a $\mu$ based on condition number in the GFS sampling which would bring some approximation error.
Table \ref{mu} shows the performance of the designed $\mu=1/(\kappa_0-1)$ approaches that of a extremely small $\mu$.
This indicates the proposed $\mu$ for stable computation won't bring a large performance gap between the original A-optimality value and the augmented one.
In Section \ref{sec:reconstruction}, we propose a new shift parameter $\beta$ based on equation \eqref{beta}.
Table \ref{mu} tells that if we apply $\mu=\beta$, the sampling performance will also be comparable.
\vspace{-0.15in}
\subsection{Static Graph Signal Reconstruction}
Simulations on evaluating the proposed GFS reconstruction algorithm are performed subsequently, where the samples were all collected by the GFS sampling algorithm.
Fig.\,\ref{recMSE} shows the reconstruction MSE of different recovery strategies in {\bf{G1}}--{\bf{G3}} at 0dB.
As illustrated in Fig.\,\ref{recMSE}, the GFS reconstruction outperforms the LS reconstruction for different graphs when the shift parameter $\beta$ is designed by \eqref{beta}, which validates the robustness of the proposed GFS reconstruction compared with the LS solution.

\begin{table}
\caption{Reconstruction MSE of The Proposed GFS Algorithm in Different Shift $\mu$ at 0dB}
\label{mu}
\begin{center}
\begin{tabular}{c|c|c|c|c|c|c}
\hline
\multirow{2}*{Graph}&\multirow{2}*{$\mu$}&\multicolumn{5}{c}{Sample size}\\
\cline{3-7}
~&~&100&110&120&130&140\\
\hline
\hline
\multirow{3}*{\textbf{G1}}&$10^{-5}$&16.10&14.55&13.43&12.44&11.63\\
\cline{2-7}
~&$1/99$&16.07&14.59&13.43&12.46&11.64\\
\cline{2-7}
~&\eqref{beta}&16.10&14.62&13.47&12.50&11.64\\
\hline
\hline
\multirow{3}*{\textbf{G2}}&$10^{-5}$&20.77&18.68&17.09&15.77&14.63\\
\cline{2-7}
~&$1/99$&20.77&18.73&17.12&15.78&14.64\\
\cline{2-7}
~&\eqref{beta}&21.31&19.36&17.59&16.33&15.08\\
\hline
\hline
\multirow{3}*{\textbf{G3}}&$10^{-5}$&23.09&20.79&18.99&17.48&16.19\\
\cline{2-7}
~&$1/99$&22.98&20.77&18.99&17.49&16.21\\
\cline{2-7}
~&\eqref{beta}&23.11&20.87&19.04&17.55&16.32\\
\hline
 \end{tabular}
 \end{center}
\end{table}

Suppose that $\beta=\mu=1/(\kappa_0-1)$, the performance of the proposed GFS sometimes achieves better performance than the LS solution.
In this case, the GFS reconstruction does not compute the matrix inverse, thus enjoys lower complexity in recovery, as remarked in Section \ref{sec:reconstruction}.
Based on the experimental results in Fig.\,\ref{recMSE} and Table \ref{mu}, we can design the shift $\mu$ in sampling and $\beta$ in reconstruction both from equation \eqref{beta}.

\vspace{-0.15in}
\subsection{Dynamic Subset Sampling}
We perform simulations to evaluate the proposed GFS-NE dynamic sampling strategy.
In our experiments, we set $J=18N\text{log}(N)$, $P_0=0.8$, $\epsilon=2\%$ and $K_0=50$.
We first randomly pick up $0.8N$ nodes to form the initial available subset, then this available set evolves over time based on the model defined in Section \ref{sec:DSS}
\footnote{
There sometimes appears a {bad} initial subset, then all sampling algorithms perform extremely bad on it, which is meaningless to do sampling work on such subset.
We propose one strategy to recognize bad subsets, see appendix \ref{badSET} for details, on which we  do not collect samples.}.
For comparison, we simulate the GFS, the E-optimal \cite{SCsampling} and SP \cite{AO} sampling algorithms, all of which select samples from scratch from new available subset.
Fig.\,\ref{dynamicMSE} shows the reconstruction MSE of different sampling strategies in terms of time evolution at 0dB with fixed sample size.
It shows that our proposed GFS-NE ourperforms SP and E-optimal and is comparable to GFS selecting samples from scratch at each time $t$ \footnote{Note that a competing greedy strategy can also replace each unavailable node in $\mathcal{S}_t$ one-by-one instead of selecting the entire set at time $t+1$ from scratch.
However, {it has to compute the initial objective value of the SA set to proceed sampling}, and its performance will be necessarily be no better than the ``start-over" approach.
Hence our method will also outperform competing schemes with unavailable node replacement. }. \begin{figure*}[htbp]
    \centering
     \subfigure[Reconstruction MSE in {\textbf{G1}}, $M=150$]{
    \begin{minipage}{5.5cm}
    \centering
        \includegraphics[width=2.0in,height=2.0in]{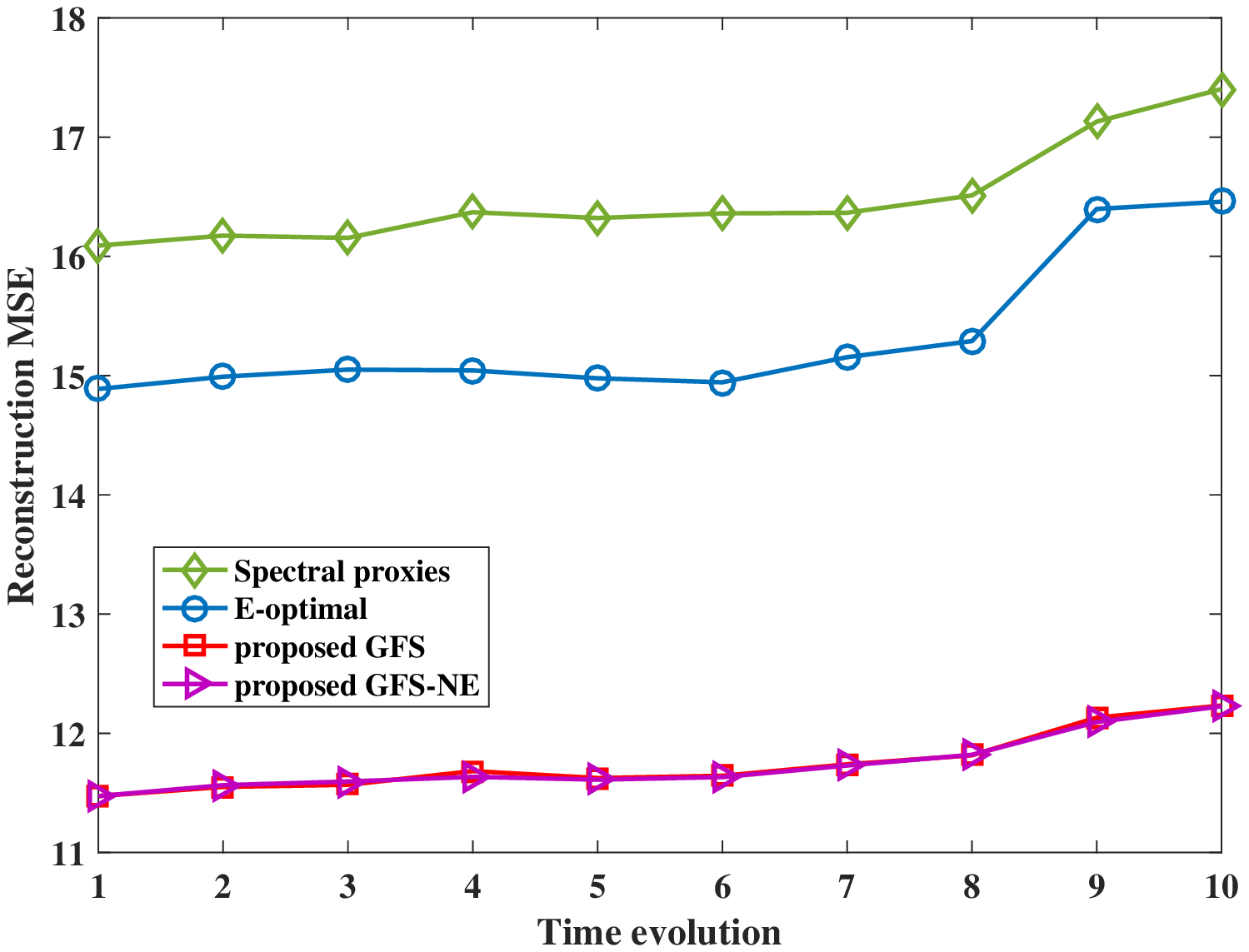}
    \end{minipage}%
    }
     \subfigure[Reconstruction MSE in {\textbf{G2}}, $M=150$]{
    \begin{minipage}{5.5cm}
    \centering
        \includegraphics[width=2.0in,height=2.0in]{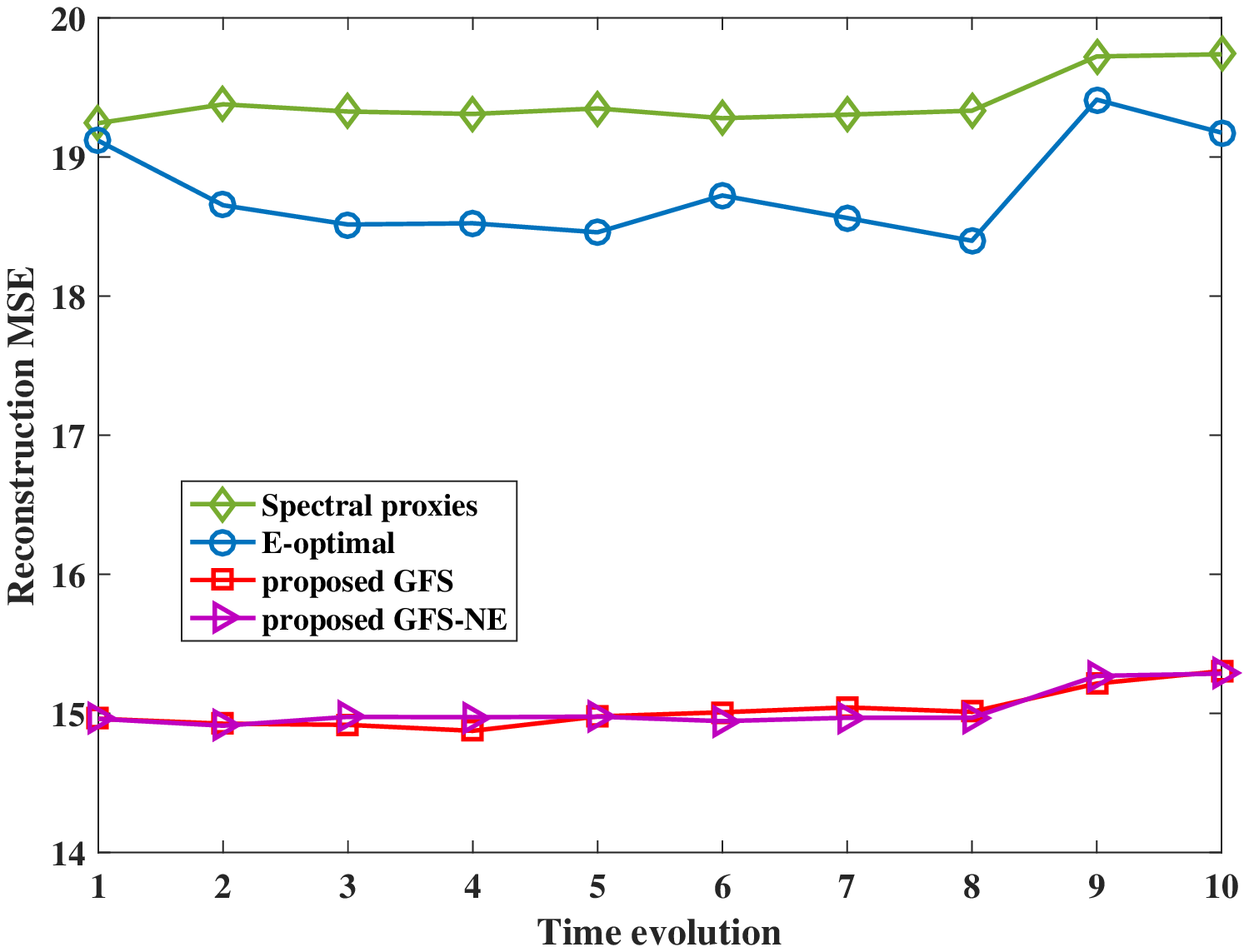}
    \end{minipage}%
    }
    \subfigure[Reconstruction MSE in {\textbf{G3}}, $M=100$]{
    \begin{minipage}{5.5cm}
    \centering
        \includegraphics[width=2.0in,height=2.0in]{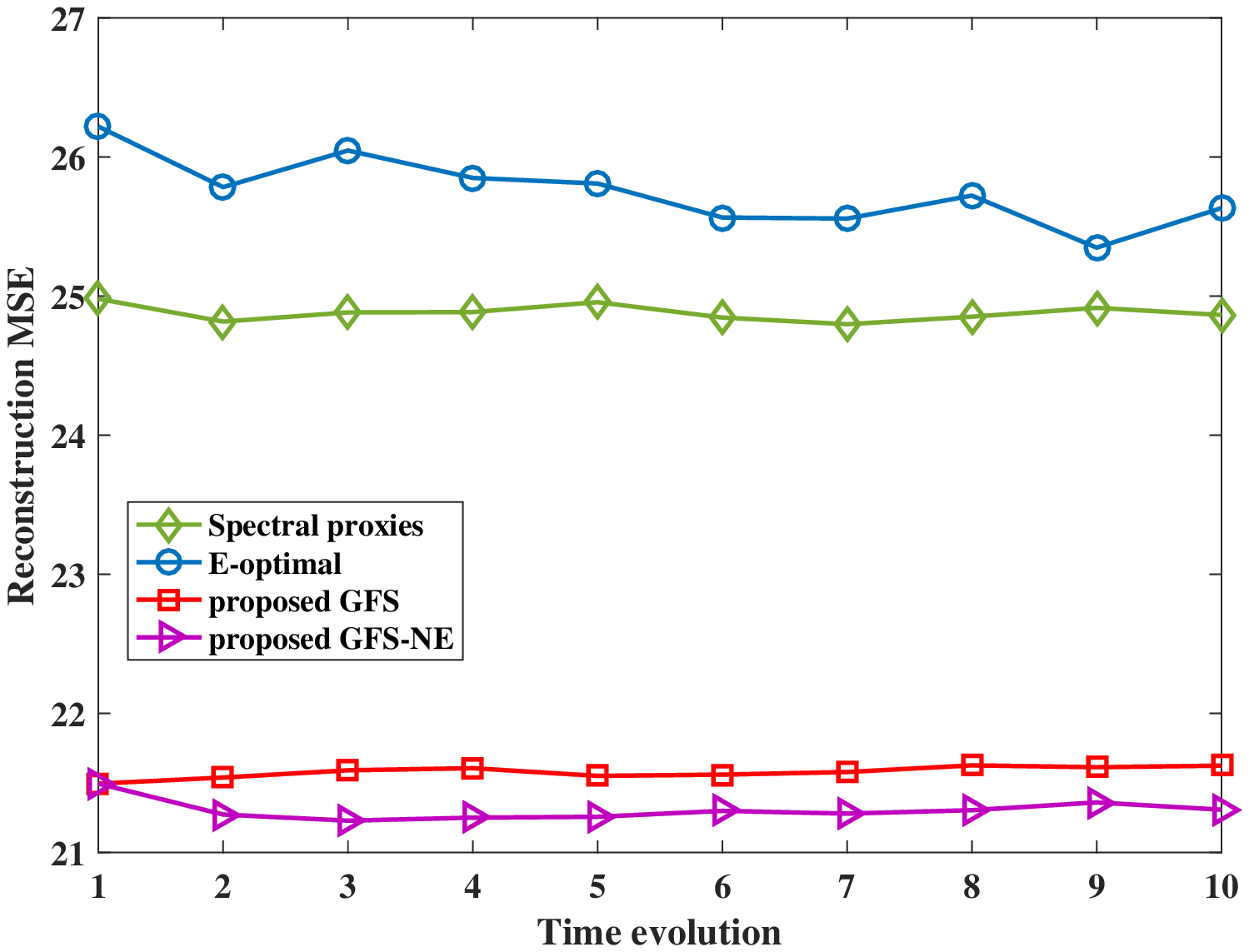}
    \end{minipage}%
    }
\caption{Experimental results of dynamic subset sampling. Different sampling algorithms are simulated in different graph types at $\textrm{SNR}=0\textrm{dB}$. The original graph signals are all recovered from selected samples by the LS reconstruction.}
\label{dynamicMSE}
\end{figure*} 

%% file: conclusion.tex
Graph sampling with noise remains a challenging problem: MMSE leads to the known A-optimality criterion for independent noise, which is expensive to evaluate and difficult to optimize.
In this paper, we propose an augmented objective
based on Neumann series expansion to approximate the A-optimality criterion, which can be expressed as a function of an ideal LP graph filter, efficiently approximated via fast graph Fourier transform.
Using the augmented objective, we select nodes greedily without any matrix inverse computation based on a matrix inversion lemma.
Further, we extend our sampling scheme to the dynamic network case, where the availability of nodes is time-varying.
For signal recovery, we design an accompanied signal reconstruction strategy to obtain a biased but robust estimator.
Experimental results validate the superiority of the proposed sampling strategy compared with existing schemes and demonstrate the effectiveness of our biased recovery algorithm over unbiased LS reconstruction. 

%% file: SPLextension.bbl
\begin{thebibliography}\label{}

\bibitem{GSP}
A. Ortega, P. Frossard, {J. Kova{\v{c}}evi{\'{c}}}, J. M. F. Moura, P. Vandergheynst, ``Graph signal processing: Overview, challenges and applications, '' \emph{Proceedings of the IEEE}, vol. 106, no. 5, pp. 808-828, May. 2018.

\bibitem{emerging}
D. Shuman, S. Narang, P. Frossard, A. Ortega, and P. Vandergheynst, ``The emerging field of signal processing on graphs: Extending high-dimensional data analysis to networks and other irregular domains,'' {\emph{IEEE Signal Process. Mag.}}, vol. 30, no. 3, pp. 83-98, May. 2013.

\bibitem{GSIP}
G. Cheung, E. Magli, Y. Tanaka, M. Ng, ``Graph Spectral Image Processing," \emph{Proceedings of the IEEE}, vol. 106, no. 5, pp. 907-930, May 2018.

\bibitem{frequency}
A. Gavili and X. P. Zhang, ``On the shift operator, graph frequency, and optimal filtering in graph signal processing,'' \emph{IEEE Trans. Signal Process.}, vol. 65, no. 23, pp. 6303-6318, Dec. 2017.

\bibitem{wavelet}
D. Hammond, P. Vandergheynst, and R. Gribonval, ``Wavelets on graphs via spectral graph theory,'' {\emph{Appl. Comput. Harmon. Anal.}}, vol. 30, no. 2, pp. 129-150, 2011.


\bibitem{filterbank}
D. B. H. Tay and A. Ortega, ``Bipartite graph filter banks: Polyphase analysis and generalization,'' \emph{IEEE Trans. Signal Process.}, vol. 65, no. 18, pp. 4833-4846, Sep. 2017.


\bibitem{gsp}
A. Sandryhaila and J. Moura, ``Big data analysis with signal processing
on graphs: Representation and processing of massive data sets with irregular structure,'' {\emph{IEEE Signal Process. Mag.}}, vol. 31, no. 5, pp. 80-90, Sep. 2014.

\bibitem{Irregularity}
B. Girault, A. Ortega, and S. Narayanan, ``Irregularity-aware graph fourier transforms'', Feb. 2018, arXiv: 1802.10220 [eess.SP].

\bibitem{direct-graph}
S. Sardellitti, S. Barbarossa, and P. Di Lorenzo, ``On the graph Fourier transform for directed graphs,'' \textit{ IEEE Journal of Selected
Topics in Signal Processing}, vol. 11, no. 6, pp. 796-811, 2017.
\bibitem{remote sensing}
H. Shen, X. L, Q. Cheng, C. Zeng, G. Yang, H. Li, and L. Zhang, ``Missing Information Reconstruction of Remote Sensing Data: A Technical Review,'' in \emph{IEEE Geoscience and Remote Sensing Magazine}, vol. 3, no. 3, pp. 61-85, Sept. 2015.

\bibitem{medical}
S. A Danziger, R. Baronio, L. Ho, L. Hall, K. Salmon, G. W. Hatfield,
P. Kaiser, and R. H. Lathrop, ``Predicting positive p53 cancer rescue
regions using most informative positive (MIP) active learning,'' \emph{PLOS
Comput. Biol.}, vol. 5, no. 9, Sept. 2009

\bibitem{cornell}
I. Shomorony and A. Avestimehr, ``Sampling large data on graphs,'' in
{\emph{Proc. IEEE Global Conf. Signal Inf. Process. (GlobalSIP)}}, pp. 933-936, Dec. 2014.


%
%
%

%

\bibitem{towards}
A. Anis, A. Gadde, and A. Ortega, ``Towards a sampling theorem
for signals on arbitrary graphs,'' in {\emph{Proc. IEEE Int. Conf. Acoust.,
Speech, Signal Process. (ICASSP)}}, Florence, Italy, pp. 3864-3868,
May. 2014.

\bibitem{AO}
A. Anis, A. Gadde, and A. Ortega, ``Efficient sampling set selection for
bandlimited graph signals using graph spectral proxies,'' {\emph{IEEE Trans.
Signal Process.}}, vol. 64, no. 14, pp. 3775-3789, Jul. 2016.

\bibitem{Pesenson2008}
I. Pesenson, ``Sampling in Paley-wiener spaces on combinatorial
graphs,'' {\emph{Trans. Amer. Math. Soc.}}, vol. 360, no. 10, pp. 5603-5627,
2008.

\bibitem{Boyd}
S. Boyd and L. Vandenberghe, {\emph{Convex Optimization}}, Cambridge University Press, 2004.

\bibitem{Uncertainty}
M. Tsitsvero, S. Barbarossa, and P. Di Lorenzo, {``Signals on graphs: Uncertainty principle and sampling,'' } {{\emph{IEEE Trans. Signal Process.}}}, vol. 64, no. 18, pp. 4845-4860, Sep.15, 2016.

\bibitem{greedybound}
L. F. O. Chamon and A. Ribeiro, {``Greedy sampling of graph signals,''} {{\emph{IEEE Trans. Signal Process.}}}, vol. 66, no. 1, pp. 34-47, Jan.1, 2018.

\bibitem{SCsampling}
S. Chen, R. Varma, A. Sandryhaila, and {J. Kova{\v{c}}evi{\'{c}}}, ``Discrete signal
processing on graphs: Sampling theory,'' {\emph{IEEE Trans. Signal Process.}},
vol. 63, no. 24, pp. 6510-6523, 2015.

\bibitem{sensorSelection}
A. Sakiyama, Y. Tanaka, T. Tanaka, and A. Ortega, ``Accelerated sensor position selection using graph localization operator,'' in {\emph{Proc. IEEE Int. Conf. Acoust., Speech, Signal Process. (ICASSP)}}, New Orleans, LA, pp. 5890-5894, Mar. 2017.

\bibitem{eigenfree}
A. Sakiyama, Y. Tanaka, T. Tanaka, and A. Ortega, {``Eigendecomposition-Free sampling set selection for graph signals,'' } Sep. 2018, arXiv:1809.01827 [eess.SP].


\bibitem{SPL}
F. Wang, Y. Wang and G. Cheung, ``A-optimal sampling and robust reconstruction for graph signals via truncated Neumann series,'' \emph{IEEE Signal Processing Letters}, vol. 25, no. 5, pp. 680-684, May. 2018.

\bibitem{randomSampling}
G. Puy, N. Tremblay, R. Gribonval, and P. Vandergheynst, ``Random sampling of bandlimited signals on graphs,''  {\emph{Appl. Comput. Harmon. Anal.}}, 2016.

\bibitem{structured sampling}
G. Puy and P. P\'{e}rez, ``Structured sampling and fast reconstruction of smooth graph signals,''
Feb. 2017, arXiv: 1705.02202 [cs.IT].

\bibitem{neumannseries}
R. A. Horn and C. R. Johnson, \emph{Matrix Analysis.} Cambridge University Press, 2013.

\bibitem{France Givens}
L. Le Magoarou, R. Gribonval and N. Tremblay, ``Approximate fast graph Fourier transforms via multilayer sparse approximations,'' in \emph{IEEE Transactions on Signal and Information Processing over Networks}, vol. 4, no. 2, pp. 407-420, Jun. 2018.

\bibitem{sensorBroken}
T. Nozaki, T. Nakano, and T. Wadayama, ``Analysis of breakdown probability of wireless sensor networks with unreliable relay nodes,'' \emph{ 2017 IEEE International Symposium on Information Theory (ISIT)}, Aachen, 2017, pp. 481-485.

\bibitem{matrix lemma}
W. W. Hager, ``Updating the inverse of a matrix,'' SIAM Rev. vol. 31 no. 2 pp. 221-239 Jun. 1989.

\bibitem{CMUicassp}
S. Chen, A. Sandryhaila, and {J. Kova{\v{c}}evi{\'{c}}}, ``Sampling theory for graph signals,'' in {\emph{Proc. IEEE Int. Conf. Acoust., Speech, Signal Process. (ICASSP)}}, South Brisbane, Queensland, Australia, pp. 3392-3396,
Apr. 2015.

\bibitem{Statistical signal}
Steven M. Kay, {\emph{Fundermentals of Statistical Signal Processing: Estimation Theory}}, Prentice-Hall, 1993.

\bibitem{experiments}
F. Pukelsheim, {\emph{Optimal Design of Experiments}}. Philadelphia, PA, USA: SIAM, 1993, vol. 50.

\bibitem{Jacobi}
G. H. Golub and H. A. Van der Vorst, ``Eigenvalue computation in the
20th century,'' \emph{Journal of Computational and Applied Mathematics}, vol.
123, no. 1, pp. 35-65, 2000.

\bibitem{condition number}
J. Demmel, \emph{Applied Numerical Linear Algebra}, U.K. Cambridge:Cambidge Univ. Press 1997.

\bibitem{blockwise}
D. S. Bernstein, \emph{Matrix Mathematics}, Princeton University Press 2005.

\bibitem{LOBPCG}
A. V. Knyazev, ``Toward the optimal preconditioned eigensolver: Locally
optimal block preconditioned conjugate gradient method,''\emph{SIAM J. Scientif. Comput.}, vol. 23, no. 2, pp. 517-541, 2001.



\bibitem{GSPtool}
N. Perraudin, J. Paratte, D. Shuman, V. Kalofolias, P. Vandergheynst, and D. K. Hammond, ``GSPBOX: A toolbox for signal processing on graphs,''
Aug. 2014, arXiv:1408.5781 [cs.IT].


\bibitem{semisupervised}
A. Gadde, A. Anis, and A. Ortega, ``Active semi-supervised learning using sampling theory for graph signals,'' in \emph{Proc. 20th ACM SIGKDD Int. Conf.
Knowl. Discov. Data Min.}, pp. 492-501, 2014.

\end{thebibliography}
